\documentclass{article}

\usepackage{microtype}
\usepackage{graphicx}
\usepackage{subfigure}
\usepackage{booktabs} % for professional tables

\usepackage{mathtools} 

\usepackage{hyperref}
\hypersetup{colorlinks,
            linkcolor=blue,
            citecolor=blue,
            urlcolor=magenta,
            linktocpage,
            plainpages=false}

\usepackage[utf8]{inputenc} % allow utf-8 input
\usepackage[T1]{fontenc}    % use 8-bit T1 fonts
\usepackage{hyperref}       % hyperlinks
\usepackage{url}            % simple URL typesetting
\usepackage{booktabs}       % professional-quality tables
\usepackage{amsfonts}       % blackboard math symbols
\usepackage{nicefrac}       % compact symbols for 1/2, etc.
\usepackage{microtype}      % microtypography
\usepackage{xcolor}         % colors

\usepackage{bm}
\usepackage{bbm}
\usepackage{comment}         
\usepackage{multicol}
\usepackage{multirow}
\usepackage{caption}
\usepackage{natbib}
\usepackage{braket}

\newcommand\sign {\ensuremath{\operatorname{\mathrm{sign}}}}

\usepackage{amsmath,amssymb,amsthm}

\usepackage{color}

\usepackage{graphicx}
\newcommand{\indep}{\rotatebox[origin=c]{90}{$\models$}}

\DeclareMathOperator*{\argmin}{arg\,min}

\usepackage{times}

\def\Holder{{H\"{o}lder}}

\newcommand{\pa}{\mathrm{\pa}}

\newcommand{\RN}[1]{%
  \textup{\uppercase\expandafter{\romannumeral#1}}%
}

\usepackage{xcolor}
\newcount\Comments  % 0 suppresses notes to selves in text
\newcommand{\kibitz}[2]{\ifnum\Comments=1\textcolor{#1}{#2}\fi}

\usepackage{algorithm}
\usepackage{algorithmic}

\def\Holder{{H\"{o}lder}}

\usepackage{authblk}
\allowdisplaybreaks

\usepackage{arxiv}

\usepackage{amsmath}
\usepackage{amssymb}
\usepackage{mathtools}
\usepackage{amsthm}

\usepackage[capitalize,noabbrev]{cleveref}

\theoremstyle{plain}
\newtheorem{theorem}{Theorem}[section]
\newtheorem{proposition}[theorem]{Proposition}
\newtheorem{lemma}[theorem]{Lemma}

\theoremstyle{definition}
\newtheorem{definition}[theorem]{Definition}
\newtheorem{assumption}[theorem]{Assumption}
\theoremstyle{remark}

\usepackage[textsize=tiny]{todonotes}

\usepackage{enumitem}

\title{Triple/Debiased Lasso for Statistical Inference of Conditional Average Treatment Effects}

\renewcommand{\shorttitle}{Triple/Debiased Lasso}

\author{Masahiro Kato}

\affil{Department of Basic Science, the University of Tokyo}

\usepackage{times}

\begin{document}

\maketitle

\begin{abstract}
This study investigates the estimation and the statistical inference about \emph{Conditional Average Treatment Effects} (CATEs), which have garnered attention as a metric representing individualized causal effects. In our data-generating process, we assume linear models for the outcomes associated with binary treatments and define the CATE as a difference between the expected outcomes of these linear models. This study allows the linear models to be high-dimensional, and our interest lies in consistent estimation and statistical inference for the CATE. In high-dimensional linear regression, one typical approach is to assume sparsity. However, in our study, we do not assume sparsity directly. Instead, we consider sparsity only in the difference of the linear models. We first use a doubly robust estimator to approximate this difference and then regress the difference on covariates with Lasso regularization. Although this regression estimator is consistent for the CATE, we further reduce the bias using the techniques in double/debiased machine learning (DML) and debiased Lasso, leading to $\sqrt{n}$-consistency and confidence intervals. We refer to the debiased estimator as the triple/debiased Lasso (TDL), applying both DML and debiased Lasso techniques. We confirm the soundness of our proposed method through simulation studies.
\end{abstract}

\section{Introduction}
The estimation of causal effects from observational data, especially in the context of binary treatments, is a crucial task in various disciplines, including economics \citep{Wager2018}, medicine \citep{Assmann2000}, and online advertising \citep{Bottou13}. Our focus is on estimating Conditional Average Treatment Effects \citep[CATEs,][]{hahn1998role,Heckman1997,Abrevaya2015}, the difference between expected outcomes of binary treatments conditioned on covariates. CATEs play an important role as a metric that represents individual heterogeneous treatment effects.

In this study, given observations generated from a distribution, we consider estimating the CATE. There are two main challenges in CATE estimation. First, we can only observe the outcomes of the treatments that are actually assigned, while the outcomes of the unassigned treatments remain unknown. This issue requires us to deal with the unobserved outcomes in some manner to estimate the CATE accurately. Second, estimating the CATE necessitates regression models to approximate the conditional expected outcomes of the treatments, which demands specific estimation methods depending on the regression models employed.

In this problem, we assume linear regression models for CATE estimation with high-dimensional covariates. We regress the covariates on estimated counterfactual differences of outcomes to estimate the CATE. To handle high dimensionality, we employ Lasso regularization, which leverages the sparsity of linear regression models.

Throughout this study, we assume sparsity for linear CATE regression models without imposing sparsity on regression models for each treatment's outcome; that is, we impose sparsity only for the difference of outcomes, not for each outcome. In this setting, we propose a three-step estimation method: (i) estimating nuisance parameters with some estimation methods with a convergence rate guarantee ; (ii) estimating the difference in outcomes and regressing covariates on them with Lasso regularization; and (iii) debiasing the bias introduced by the use of Lasso regularization using the debiased Lasso. The first and second steps depend on double/debiased machine learning (DML) technique \citep{klaassen1987,ZhengWenjing2011CTME,ChernozhukovVictor2018Dmlf}, while the third step is based on the debiased Lasso \citep{vandeGeer2014}. Therefore, we refer to our estimation method as the Triple/Debiased Lasso (TDL) for CATE estimation. We demonstrate the obtained estimator's consistency and asymptotic normality.

The challenge in parameter estimation in high-dimensional regression scenarios, particularly when the parameter space exceeds the sample size, is significant. Our approach adopts high-dimensional but sparse linear regression models for the CATE and applies Lasso regression \citep{tibshirani96regression, zhao06a, vandeGeer2008}. Following existing studies, such as \citet{kato2023cate}, we do not assume sparsity for regression models between covariates and each treatment's outcome. Instead, we directly assume sparsity for the regression models of the difference between the outcomes.

Since we cannot observe the outcomes of the unassigned treatments, we approximate them with a doubly robust (DR) estimator \citep{bang2005drestimation} with nuisance parameters estimated via the cross-fitting in DML. This method guarantees the first convergence of the bias caused by the DR estimation. 

Then, we regress covariates on the estimated difference with Lasso regularization. Since Lasso regularization also yields bias, we debias it using the debiased (desparsified) Lasso technique \citep{vandeGeer2014}. This technique leads us to construct a confidence interval for each regression coefficient. 

Furthermore, in Lasso regression, we utilize the weighted least squares to reduce the asymptotic variance of the regression coefficient estimators due to heterogeneous variances.

The use of DR estimation prior to regression has existing studies, such as \citet{vanderLaan2006}, \citet{Lee2017}, \citet{ninomiya2021selective}, \citet{Fan2022}, and \citet{tan2024combining}. Such approaches are summarized as the DR-learner by \citet{Kunzel2019}. Among methods categorized as variants of the DR-learner, \citet{Fan2022} and \citet{ninomiya2021selective} consider DR estimator with confidence intervals. While \citet{Fan2022} employs the nonparametric regression with multiplier-bootstrap, \citet{ninomiya2021selective} assumes linear models and employs Lasso regression with the selective inference. Our method is also categorized into the DR-learner with linear models and DML, but we construct a confidence interval using the debiased Lasso.

Our basic estimation strategy is similar to that in \citet{Fan2022}, which also employs DML to debias the bias caused by the first-step machine learning estimator. The difference between our estimator and their estimator lies in the use of weighted least squares and debiased Lasso. While \citet{Fan2022} does not take into account the existence of the heterogeneous variances, the conditional variance given covariates differ across the covariates due to the covariate-dependent propensity score; that is, the variance is heterogeneous, and under the standard arguments in linear regression, we can minimize the asymptotic variance of regression coefficient estimators by using the weighted least squares. Following this argument, we add the weighted least squares between the procedures of DML and inference, which has not been incorporated in \citet{Fan2022}. Furthermore, we utilize the debiased Lasso to obtain a confidence interval, following \citet{vandeGeer2014}, in contrast to the multiplier-bootstrap in \citet{Fan2022}. 

\textbf{Organization.} Our problem setting is formulated in Section~\ref{sec:problem}. For linear regression models, we introduce Lasso-based estimators with the DR estimation, DML, and debiased Lasso in Section~\ref{sec:cateLasso}, called the TDL estimator. For the TDL estimator, we show the consistency and asymptotic normality in Section~\ref{sec:theoretical}. Section~\ref{sec:related} introduces related work.

\section{Problem Setting}
\label{sec:problem}
This section formulates the problem of CATE estimation. In Section~\ref{sec:potentialoutcome}, potential outcomes are introduced, following the Neyman-Rubin causal model \citep{Neyman1923,Rubin1974}. In Section~\ref{sec:obs}, the observations are defined. Subsequently, in Section~\ref{sec:cate}, the CATE, the target of our estimation using the observations, is defined.

\subsection{Potential Outcomes}
\label{sec:potentialoutcome}
In our problem, two treatments $\{1, 0\}$ are given.\footnote{For example, we often regard treatment $1$ as a treatment of interest, and treatment $0$ as the control treatment.} Each treatment yields a corresponding outcome. Following the Neyman-Rubin causal model \citep{Neyman1923,Rubin1974}, for each treatment $d\in\{1, 0\}$, we introduce potential outcome random variables $Y(d) \in\mathbb{R}$. We assume that $Y(1)$ and $Y(0)$ are bounded random variables.
\begin{assumption}
    \label{asm:bounded_output}
    For each $d\in\{1, 0\}$, $Y(d)$ is bounded.
\end{assumption}

We characterize individuals who receive treatments by $p$-dimensional covariates $X\in\mathcal{X}$, where $\mathcal{X}\subset\mathbb{R}^p$ is the covariate space. 

\subsection{Observations}
\label{sec:obs}
In the previous section, we introduced the potential outcomes $Y(1)$ and $Y(0)$. However, we can only observe one of the outcomes corresponding to a treatment actually assigned, which is denoted by $D\in\{1, 0\}$. Let $D \in \{1, 0\}$. Then, an observable outcome $Y\in\mathbb{R}$ is defined as  
\begin{align*}
Y = D Y(1) + (1 - D) Y(0),
\end{align*}
Here, if $D = 1$, $Y = Y(1)$ holds; if $D = 0$, $Y = Y(0)$ holds.

Let $P_0\in\mathcal{P}$ be the ``true'' distribution that generates $(Y, D, X)$, where $\mathcal{P}$ is the set of all possible distributions. Let us denote the sample size by $n\in\mathbb{N}$. For each $i \in [n] \coloneqq \{1,2,\dots,n\}$, we define a triple $(Y_i, D_i, X_i)$ as an independent and identically distributed (i.i.d.) copy of $(Y, D, X)$.
Then, we define our observations as the following dataset $\mathcal{D}$:
\begin{align*}
\mathcal{D} \coloneqq \left\{(Y_i, D_i, X_i)\right\}^n_{i=1}.
\end{align*}

\subsection{CATE}
\label{sec:cate}
For $P\in\mathcal{P}$, let $\mathbb{E}_P[W]$, $\mathrm{Var}_P[W]$, and $\mathbb{P}_P[W]$ be the expectation operator, the variance operator, and the probability law a distribution $P$ of a random variable $W$. For each $d\in\{1, 0\}$ and any $x\in\mathcal{X}$, let us denote the conditional expected outcome of $Y(d)$ by $\mu_P(1)(x) \coloneqq  \mathbb{E}_P[Y(d)\mid X = x]$ and the propensity score by  $\pi_0(d\mid x) \coloneqq \mathbb{P}_{P_0}[D = d\mid X]$. 

In this study, our interest lies in CATEs at $X = x \in \mathcal{X}$ under $P_0$, defined as
\begin{align*}
    f_0(x) \coloneqq f_{P_0}(x) \coloneqq  \mathbb{E}_{P_0}[Y(1) - Y(0) \mid X = x] = \mu_0(1)(x) - \mu_0(0)(x).
\end{align*}
This quantity has been widely used in empirical studies of various fields, such as epidemiology, economics, and political science, because it captures heterogeneous treatment effects for each individual represented by a characteristic $x\in\mathcal{X}$. 

For identification of $f_0(x)$, we make the following assumptions, called the \emph{unconfoundedness} and the \emph{overlap of assignment support}, respectively.
\begin{assumption}[Unconfoundedness] \label{asmp:unconfounded}
The potential outcomes $(Y(1), Y(0))$ are independent of the treatment indicator $D$ conditional on $X$:
\[(Y(1), Y(0))\indep ~D \ |\ X\]
\end{assumption}

\begin{assumption}[Overlap of assignment support] \label{asmp:coherent}
For some universal constant $0 < \varphi < 0.5$, for each $d\in\{1, 0\}$ we have
\begin{align*}
    \varphi < \pi_0(d\mid X) < 1 - \varphi,\qquad \mathrm{a.s.}
\end{align*}
\end{assumption}

Under the Assumption \ref{asmp:unconfounded}, assignment $D$ is independent of potential outcomes $(Y(1), Y(0))$ conditioned on covariates.
Assuming the unconfoundedness is the standard approach in treatment effect estimation \citep{Rosenbaum1983}. Assumption~\ref{asmp:coherent} is made to avoid disjoint support between treatment and control groups.
These assumptions guarantee the identifiability of the CATE \citep{imbens_rubin_2015}.

\paragraph{Notation.} 
Let $[p] \coloneqq \{1,\dots,p\}$. 
Let us define a ($n\times p$)-matrix 
$
   \bm{X} \coloneqq \begin{pmatrix}
    X^\top_{1} \\
    X^\top_{2} \\
    \vdots \\
    X^\top_{n} 
\end{pmatrix} = \begin{pmatrix}
    X_{1, 1} & X_{1, 2} & \cdots & X_{1, p}\\
    X_{2, 1} & X_{2, 2} & \cdots & X_{2, p}\\
    \vdots & \vdots & \ddots & \vdots\\
    X_{n, 1} & X_{n, 2} & \cdots & X_{n, p}
\end{pmatrix}$,  and a $n$-dimensional vector $
\mathbb{D} \coloneqq (D_1,D_2,...,D_n)^\top
$.

Let us define $n$-dimensional column vectors ${\mathbb{Y}}$ and ${\mathbb{X}}_j$ for each $j\in[p]$ as 
$
    {\mathbb{Y}} = (Y_1\ Y_2\ \cdots\ Y_n)^\top$, $
    {\mathbb{X}}_j = (X_{1,j}\ X_{2,j}\ \cdots\ X_{n,j})^\top_{i\in[n]}
$. 
Here, we can also denote $\bm{X}$ by $    \bm{X} = \begin{pmatrix}
        {\mathbb{X}}_1 & {\mathbb{X}}_2 & \cdots & {\mathbb{X}}_p
    \end{pmatrix}$. 
For a vector $v = (v_1,\dots, v_p)\in\mathbb{R}^p$, let us denote $\ell_q$ norm by $\|v\|_q \coloneqq (\sum^p_{i=1}|v_i|^q)$ for $q \geq 1$. Let $\mathcal{M}$ be the set of measurable functions $\mu:\{1, 0\}\times \mathcal{X}\to \mathbb{R}$, and $\Pi$ be the set of measurable functions $\pi:\{1, 0\}\times \mathcal{X} \to [0, 1]$ such that $\pi(1\mid x) + \pi(0\mid x) = 1$ for any $x\in\mathcal{X}$. Let $\mathrm{Var}(\cdot)$ be the variance operator. 

\section{TDL for the CATE Estimation}
\label{sec:cateLasso}
This section introduces high-dimensional linear regression models with sparsity for estimating the CATE.

\subsection{High-Dimensional Linear Regression with Sparsity}
\label{sec:ind_common}
In this study, under $P\in\mathcal{P}$, we posit the following linear regression models between an individual treatment effect $Y(1) - Y(0)$ and covariates $X$:
\begin{align}
\label{eq:linear}
    Y(1) - Y(0) = X^\top\bm{\beta} + \varepsilon(d), 
\end{align}
where $\bm{\beta} \in \mathbb{R}^p$ is a $p$-dimensional parameter, while $\varepsilon$ is an independent noise variable with a zero mean and finite variance. Let us define $\varepsilon(d) \coloneqq Y(d) - \mu_P(d)(X)$ for each $d\in\{1, 0\}$. In our analysis, we make the following assumptions among $\varepsilon$, $\varepsilon^1$, and $\varepsilon^0$. 
\begin{assumption}
\label{asm:eror}
For the error term $\epsilon$, $\mathbb{E}_{P_0}[\varepsilon \mid X] = \mathbb{E}_{P_0}[\varepsilon^1 - \varepsilon^0 \mid X] = 0$ holds a.s. Furthermore, we assume that $\varepsilon^1$ and $\varepsilon^0$ are independent. Additionally, the variance is 
finite; that is, $\sigma^2_{\bar{\varepsilon}} \coloneqq \mathbb{E}_{P_0}[(\varepsilon^1 - \varepsilon^0)^2] < C < \infty$ for some universal constant $C$. 
\end{assumption}

Let $\varepsilon_i$ and $\varepsilon_i(d)$ be i.i.d. copies of $\varepsilon$ and $\varepsilon(d)$, and $\bm{\varepsilon}$ and $\bm{\varepsilon}(d)$ be $(\varepsilon_1\ \cdots\ \varepsilon_n)^\top$ and $(\varepsilon(d)_1\ \cdots\ \varepsilon(d)_n)^\top$, respectively. 

We also allow for our linear regression models to be high-dimensional; that is, $p > n$. In such a situation, a common approach to obtaining a consistent estimator is to assume that $\bm{\beta}$ has \textit{sparsity}, that is, most of the elements of $\bm{\beta}$ are zero.

Let $\bm{\beta}_0$ be the \emph{true} parameter under $P_0$. We denote the active set of regression coefficients by 
\[\mathcal{S}_0 \coloneqq \{j \in [p] : \beta_{0, j} = 0\},\]
where $\beta_{0, j}$ ($j\in[p]$) is $j$-th element of $\bm{\beta}_0$. 
Let $s_0 \coloneqq |\mathcal{S}_0|$.

\subsection{Weighted TDL for CATE Esimation with Lasso}
\label{sec:weight_tdl}
This section introduces our proposed estimator, referred to as the Weighted TDL (WTDL) estimator for CATE estimation with Lasso. Our estimator comprises the following components:
\begin{enumerate}
\item Estimation of nuisance parameters ($\mu_0$ and $\pi_0$) using cross-fitting;
\item Application of weighted least squares to estimate the CATE regression coefficients $\bm{\beta}_0$ employing the Lasso regularization;
\item Debiasing the bias introduced by the weighted least squares with Lasso through the debiased Lasso technique.
\end{enumerate}
Our methodology incorporates DML and the debiased Lasso, leading us to term our approach TDL. The derivation and detailed explanation of the WTDL are presented in Appendix~\ref{sec:det_est_strategy}.

\subsection{Approximation of the Difference $Y(1) - Y(0)$ using Cross-Fitting}
\label{sec:diff_approx}
\paragraph{Estimation with the true nuisance parameters ($\mu_0$ and $\pi_0$).}
As well as the standard DR estimator \citep{bang2005drestimation}, we first define a function $Q:\mathcal{Y}\times \{1, 0\}, \times \mathcal{X}\times \mathcal{M} \times \Pi \to \mathbb{R}$ as
\begin{align*}
    Q(Y, D, X; \mu, \pi) \coloneqq \frac{\mathbbm{1}\big[D = 1\big]\big(Y - \mu(1)(x)\big)}{\pi(1\mid x)} - \frac{\mathbbm{1}\big[D = 0\big]\big(Y - \mu(0)(x)\big)}{\pi(0\mid x)} + \mu(1)(x) - \mu(0)(x).
\end{align*}
Here, $Q(Y, D, X; \mu_0, \pi_0)$ works as an unbiased estimator of $Y(1) - Y(0)$. Here, the following linear regression model holds under $P_0$:
\begin{align}
\label{eq:cate_linear_model}
   \overline{\mathbb{Q}} = \bm{X}\bm{\beta}_0 + \overline{\bm{\varepsilon}},
\end{align}
where $\overline{\bm{\varepsilon}}$ is an unobservable $n$-dimensional vector whose $i$-th element $\overline{\varepsilon}_i$ is equal to 
\begin{align*}
    \overline{\varepsilon}_i &\coloneqq \overline{u}_i + \varepsilon_i(1) - \varepsilon_i(0),\\
    \overline{u}_i &\coloneqq \frac{\mathbbm{1}\big[D_i \coloneqq 1\big]\big(Y_i - \mu_0(1)(X_i)\big)}{\pi_0(1\mid X_i)} - \frac{\mathbbm{1}\big[D_i = 0\big]\big(Y_i - \mu_0(0)(X_i)\big)}{\pi_0(0\mid X_i)} - \Big\{Y_i(1) - \mu(1)(X_i)\Big\} + \Big\{Y_i(0) - \mu(0)(X_i)\Big\}.
\end{align*}
Note that $\mathbb{E}_{P_0}\left[\overline{\varepsilon}_i \mid X_i\right] = 0$ since
\begin{align*}
    \mathbb{E}_{P_0}\left[\overline{\varepsilon}_i \mid X_i\right] = \mathbb{E}_{P_0}\left[\frac{\mathbbm{1}\big[D_i \coloneqq 1\big]\big(Y_i - \mu_0(1)(X_i)\big)}{\pi_0(1\mid X_i)} - \frac{\mathbbm{1}\big[D_i = 0\big]\big(Y_i - \mu_0(0)(X_i)\big)}{\pi_0(0\mid X_i)}\mid X_i\right] = 0.
\end{align*}
Let us also define 
\begin{align*}
    \sigma^2_{\bar{\varepsilon}}(x) \coloneqq \mathrm{Var}_{P_0}\big[\overline{\varepsilon}_i  \mid X_i = x\big] = \mathbb{E}_{P_0}\big[\overline{\varepsilon}^2_i  \mid X_i = x\big] = \frac{\sigma^2_{\varepsilon}(1)(x)}{\pi(1\mid x)} + \frac{\sigma^2_{\varepsilon}(0)(x)}{\pi(0\mid x)},
\end{align*}
where $\sigma^2_{\varepsilon}(d)(x)$ denotes the variance of $\varepsilon(d)$ conditioned on $x \in \mathcal{X}$.

\paragraph{Estimation with estimated nuisance parameters.}
Since we do not know $\mu_0$ and $\pi_0$, we replace them with their estimators, which cause bias in the difference estimator $Q(Y, D, X; \mu_0, \pi_0)$. We employ the cross-fitting in DML to debias the bias caused in the estimation of $\mu_0$ and $\pi_0$. The cross-fitting is a variant of the sample-splitting technique that has been utilized in the semiparametric analysis, such as ..., and refined by \citet{ChernozhukovVictor2018Dmlf}. In this technique, we first split samples into $m$ subgroups with sample sizes $n_\ell = n / m$. For simplicity, without loss of generality, we assume that $n/m$ is an integer. For each $\ell \in [m]$, we define $\mathcal{D}^\ell \coloneqq \left\{(Y_{\ell, (i)}, D_{\ell, (i)}, X_{\ell, (i)})\right\}^{n_\ell}_{i=1}$ and $\overline{\mathcal{D}}^\ell \coloneqq \mathcal{D} \backslash \mathcal{D}^\ell$, where $(Y_{\ell, (i)}, D_{\ell, (i)}, X_{\ell, (i)})$ denotes an $i$-th element in the $\ell$-th subgroup. Let $\widehat{\mu}_{\ell, n}$ and $\widehat{\pi}_{\ell, n}$ be some estimators constructed only using $\overline{\mathcal{D}}^\ell$, without using $\mathcal{D}^\ell$. The estimators are assumed to satisfy a convergence rate condition described in Assumption~\ref{asm:conv_rate}. Then, we construct a difference estimator as
\begin{align}
\label{eq:dml_est}
    \widehat{Q}^{\mathrm{DML}}_i \coloneqq Q\big(Y_{\ell(i), i}, D_{\ell(i), i}, X_{\ell(i), i}; \widehat{\mu}_{\ell(i), n}, \widehat{\pi}_{\ell(i), n}\big)\ \ \mathrm{and}\ \ \widehat{\mathbb{Q}}^{\mathrm{DML}} \coloneqq \Big(\widehat{Q}^{\mathrm{DML}}_1\ \widehat{Q}^{\mathrm{DML}}_2\ \cdots \widehat{Q}^{\mathrm{DML}}_n\Big)^\top,
\end{align}
where $\ell(i)$ denotes a subgroup $\ell \in [m]$ such that sample $i$ belongs to it. 
Here, the following linear regression model holds under $P_0$:
\begin{align*}
   \widehat{\mathbb{Q}}^{\mathrm{DML}} = \bm{X}\bm{\beta}_0 + \widehat{\bm{\varepsilon}}^{\mathrm{DML}},
\end{align*}
where $\widehat{\bm{\varepsilon}}^{\mathrm{DML}}$ is an unobservable $n$-dimensional vector whose $i$-th element $\widehat{{\varepsilon}}^{\mathrm{DML}}_i$ is equal to 
\begin{align*}
     \widehat{{\varepsilon}}^{\mathrm{DML}}_i &\coloneqq \widehat{u}^{\mathrm{DML}}_i + \varepsilon_i(1) - \varepsilon_i(0),\\
     \widehat{u}^{\mathrm{DML}}_i & \coloneqq \frac{\mathbbm{1}\big[D_i = 1\big]\big(Y_i - \widehat{\mu}_{\ell(i), n}(1)(X_i)\big)}{\widehat{\pi}_{\ell(i), n}(1\mid X_i)} - \frac{\mathbbm{1}\big[D_i = 0\big]\big(Y_i - \widehat{\mu}_{\ell(i), n}(0)(X_i)\big)}{\widehat{\pi}_{\ell(i), n}(0\mid X_i)}\\
     &\ \ \ \ \ \ \ \ \  - \Big\{Y_i(1) - \widehat{\mu}_{\ell(i), n}(1)(X_i)\Big\}+ \Big\{Y_i(0) - \widehat{\mu}_{\ell(i), n}(0)(X_i)\Big\} 
\end{align*}
In Lemma~\ref{lem:conv_Q}, we show that for each $i\in[n]$, $\sqrt{n}| \widehat{u}^{\mathrm{DML}}_i - \overline{u}_i \mid = o_P(1)$ holds, which plays an important role for deriving a confidence interval. We refer to this estimator as the DML-DR CATELasso (DML-CATELasso) estimator. Note that this estimator does not have the asymptotic normality due to the bias caused by the use of the Lasso regularization. Furthermore, if we are only interested in consistency, we do not have to apply the cross-fitting.

\subsection{Weighted Least Squares with the Lasso Regularization}
\label{sec:wtdl_reg}
We later show that the mean of $\widehat{{\varepsilon}}^{\mathrm{DML}}_i$ conditioned on $\bm{X}_i$ converges to zero as $n\to\infty$, and the variance of $\widehat{{\varepsilon}}^{\mathrm{DML}}_i$ conditioned on $\bm{X}_i$ converges to 
\begin{align*}
    \sigma^2_{\bar{\varepsilon}}(x) \coloneqq \mathrm{Var}_{P_0}\big[\overline{\varepsilon}_i  \mid X_i = x\big] = \mathbb{E}_{P_0}\big[\overline{\varepsilon}^2_i  \mid X_i = x\big] = \frac{\sigma^2_{\varepsilon}(1)(x)}{\pi(1\mid x)} + \frac{\sigma^2_{\varepsilon}(0)(x)}{\pi(0\mid x)},
\end{align*}
where $\sigma^2_{\varepsilon}(d)(x)$ denotes the variance of $\varepsilon(d)$ conditioned on $x \in \mathcal{X}$. Due to the dependency of $\sigma^2_{\bar{\varepsilon}}(x)$ on $x\in\mathcal{X}$, we can consider the weighted least squares using $\sigma^2_{\bar{\varepsilon}}(x)$ to reduce the variance of estimators $\bm{\beta}_0$. By using the weighted least sqaures and the Lasso regularization, we estimate $\bm{\beta}_0$ as
\begin{align}
\label{eq:WDML-CATELasso}
&\widehat{\bm{\beta}}^{\mathrm{WDML}}_n \in \argmin_{\bm{\beta}\in\mathbb{R}^p} \left\{\left\|\widehat{\mathbb{Q}}^{\mathrm{WDML}} - \bm{W}\bm{\beta} \right\|^2_2 / n+ \lambda\left\|\bm{\beta}\right\|_1\right\},
\end{align}
where 
\begin{align*}
    \widehat{\mathbb{Q}}^{\mathrm{WDML}} & \coloneqq  \Big( \widehat{Q}^{\mathrm{DML}}_1 / \widehat{\sigma}_{\bar{\varepsilon}}(\bm{X}_1) \ \cdots\ \widehat{Q}^{\mathrm{DML}}_n / \widehat{\sigma}_{\bar{\varepsilon}}(\bm{X}_n) \Big)^\top\\
    \bm{W} &\coloneqq \Big( \bm{X}_1 / \widehat{\sigma}_{\bar{\varepsilon}}(\bm{X}_1) \ \cdots\ \bm{X}_n / \widehat{\sigma}_{\bar{\varepsilon}}(\bm{X}_n) \Big)^\top,
\end{align*}
and $\widehat{\sigma}^2_{\bar{\varepsilon}}$ is a consistent estimator of $\sigma^2_{\bar{\varepsilon}}(x)$. We refer to this estimator as the Weighted DML-DR CATELasso (WDML-CATELasso) estimator. Let us also define 
\begin{align*}
    \widehat{\bm{\varepsilon}}^{\mathrm{WDML}} & \coloneqq  \Big( \widehat{{\varepsilon}}^{\mathrm{DML}}_1 / \widehat{\sigma}_{\bar{\varepsilon}}(\bm{X}_1) \ \cdots\ \widehat{{\varepsilon}}^{\mathrm{DML}}_n / \widehat{\sigma}_{\bar{\varepsilon}}(\bm{X}_n) \Big)^\top,\\
    \overline{\bm{\varepsilon}}^{\mathrm{Weight}} & \coloneqq  \Big( \overline{\varepsilon}_1 / \widehat{\sigma}_{\bar{\varepsilon}}(\bm{X}_1) \ \cdots\ \overline{\varepsilon}_n / \widehat{\sigma}_{\bar{\varepsilon}}(\bm{X}_n) \Big)^\top,\\
    \overline{\bm{v}} & \coloneqq  \Big( \widetilde{{u}}_1 / \widehat{\sigma}_{\bar{\varepsilon}}(\bm{X}_1) \ \cdots\ \widetilde{{u}}_n / \widehat{\sigma}_{\bar{\varepsilon}}(\bm{X}_n) \Big)^\top,\\
    \widehat{\bm{v}} & \coloneqq  \Big( \widehat{{u}}^{\mathrm{DML}}_1 / \widehat{\sigma}_{\bar{\varepsilon}}(\bm{X}_1) \ \cdots\ \widehat{{u}}^{\mathrm{DML}}_n / \widehat{\sigma}_{\bar{\varepsilon}}(\bm{X}_n) \Big)^\top.
\end{align*}

\subsection{Debiased Lasso for the WDML-CATELasso Estimator} 
The WDML-CATELasso estimator has a bias caused by Lasso. We then debias the bias using the debiased Lasso.

The estimator $\widehat{\bm{\beta}}^{\mathrm{WDML}}_n$ satisfies the following KKT conditions:
\begin{align}
\label{eq:WDML-CATELasso_kkt}
&\bm{W}^\top\left(\widehat{\mathbb{Q}}^{\mathrm{WDML}} - \bm{W}\widehat{\bm{\beta}}^{\mathrm{WDML}}_n\right) / n = \lambda \widehat{\bm{\kappa}}^{\mathrm{WDML}},
\end{align}
where $\widehat{\bm{\kappa}}^{\mathrm{WDML}} \coloneqq \sign\left(\widehat{\bm{\beta}}^{\mathrm{WDML}}_n\right) = \Big(\sign\left(\widehat{\beta}^{\mathrm{WDML}}_{n, 1}\right),\cdots, \sign\left(\widehat{\beta}^{\mathrm{WDML}}_{n, p}\right)\Big)^\top$, and for each $j\in[p]$, 
\begin{align}
\label{eq:KKTcond}
    \widehat{\kappa}^{\mathrm{WDML}}_j = \sign\left(\widehat{\beta}^{\mathrm{WDML}}_{n, j}\right) \in \begin{cases}
    1 & \mathrm{if}\quad \widehat{\beta}^{\mathrm{WDML}}_{n, 1} > 0\\
    [-1, 1] & \mathrm{if}\quad \widehat{\beta}^{\mathrm{WDML}}_{n, 1} = 0\\
    - 1 & \mathrm{if}\quad \widehat{\beta}^{\mathrm{WDML}}_{n, 1} < 0
    \end{cases}.
\end{align}

Let $\widehat{\Sigma} = \bm{X}^\top\bm{X} / n$. Let us also define $\widehat{\Theta}$ as a reasonable approximation for the inverse of $\widehat{\Sigma}$. Then, we consider debiasing the bias caused by the Lasso regularization in the WDML-CATELasso estimator to obtain a confidence interval about $\bm{\beta}_0$. From the KKT condition in \eqref{eq:WDML-CATELasso_kkt}, we have 
\begin{align*}
&\widehat{\Sigma} \left(\widehat{\bm{\beta}}^{\mathrm{WDML}}_n - \bm{\beta}_0\right)
+ \lambda \widehat{\bm{\kappa}}^{\mathrm{WDML}} =  \bm{W}^\top \widehat{\bm{\varepsilon}}^{\mathrm{WDML}} / n,
\end{align*}
where we used $\widehat{\mathbb{Q}}^{\mathrm{WDML}} = \bm{W}\bm{\beta}_0 + \widehat{\bm{\varepsilon}}^{\mathrm{WDML}}$, and 
\begin{align*}
&\bm{W}^\top\left(\widehat{\mathbb{Q}}^{\mathrm{WDML}}  - \bm{W}\widehat{\bm{\beta}}^{\mathrm{WDML}}_n\right)/n = \bm{W}^\top\left(\bm{W}\bm{\beta}_0 + \widehat{\bm{\varepsilon}}^{\mathrm{WDML}}- \bm{W}\widehat{\bm{\beta}}^{\mathrm{WDML}}_n\right) / n = \widehat{\Sigma}\left(\bm{\beta}_0 - \widehat{\bm{\beta}}^{\mathrm{WDML}}_n\right) + \bm{X}^\top\widehat{\bm{\varepsilon}}^{\mathrm{WDML}}/ n.
\end{align*}
Then, it holds that
\begin{align*}
\widehat{\bm{\beta}}^{\mathrm{WDML}}_n + \widehat{\Theta} \lambda \widehat{\bm{\kappa}}^{\mathrm{WDML}} - \bm{\beta}_0 & = \widehat{\Theta} \bm{X}^\top \widehat{\bm{\varepsilon}}^{\mathrm{WDML}} / n - \widehat{\Delta}^{\mathrm{WDML}} / \sqrt{n},
\end{align*}
where $\widehat{\Delta}^{\mathrm{WDML}} \coloneqq \sqrt{n}\widehat{\Theta}\widehat{\Sigma} \left( \widehat{\bm{\beta}}^{\mathrm{WDML}}_n - \bm{\beta}_0\right)$, which is the bias term that is expected to approaches zero as $n \to \infty$. Then, we define the Weighted TML (WTML) estimator as
\begin{align}
\label{eq:wtdl_cate_Lasso}
\widehat{\bm{b}}^{\mathrm{WTDL}}_n
&\coloneqq \widehat{\bm{\beta}}^{\mathrm{WDML}}_n + \widehat{\Theta} \lambda \widehat{\bm{\kappa}}^{\mathrm{WDML}} = \widehat{\bm{\beta}}^{\mathrm{WDML}}_n + \widehat{\Theta} \bm{W}^\top\left(\widehat{\mathbb{Q}}^{\mathrm{WDML}} - \bm{W}\widehat{\bm{\beta}}^{\mathrm{WDML}}_n\right) / n.
\end{align}
Here, it holds that 
\begin{align*}
\sqrt{n}\left(\widehat{\bm{b}}^{\mathrm{WTDL}}_n - \bm{\beta}_0\right) = \widehat{\Theta} \bm{W}^\top \widehat{\bm{\varepsilon}}^{\mathrm{WDML}} / \sqrt{n} - \widehat{\Delta}^{\mathrm{WDML}} + o_P(1),
\end{align*}

\paragraph{Nodewise Lasso regression for constructing $\widehat{\Theta}$.}
The remaining problem is the construction of an approximate inverse $\widehat{\Theta}$. Since $\widehat{\Theta}$ is necessary to construct confidence intervals, its construction has been well explored by existing studies, such as \citet{vandeGeer2014}, \citet{Javanmard14a}, \citet{Cai2017}, and \citet{Javanmard2018}. Our formulation follows the arguments in  \citet{Meinshausen2006} and \citet{vandeGeer2014}. 

Let $\bm{W}_{(-j)}$ be the design submatrix excluding the $j$-th column. Let us also define $n$-dimensional column vectors ${\mathbb{W}}_j$ for each $j\in[p]$ as ${\mathbb{X}}_j = (X_{1,j}/\widehat{\sigma}_{\bar{\varepsilon}}(\bm{X}_1)\ X_{2,j}/\widehat{\sigma}_{\bar{\varepsilon}}(\bm{X}_2)\ \cdots\ X_{n,j}/\widehat{\sigma}_{\bar{\varepsilon}}(\bm{X}_n))^\top$. 

In our construction, for each $j\in[p]$, we carry out the nodewise Lasso regression as follows:
\begin{align*}
\widehat{\bm{\gamma}}_j = \argmin_{\bm{\gamma}\in \mathbb{R}^{p-1}} &\Bigg\{\left\| \mathbb{W}_j - \bm{W}_{(-j)} \bm{\gamma} \right\|^2_2 /n + 2\lambda_j \|\bm{\gamma}\|_1 \Bigg\}, 
\end{align*}
with components of $\widehat{\bm{\gamma}}_j = \{\widehat{\gamma}_{j,k}; k \in [p], k\neq j\}$. Denote by 
\begin{align*}
\widehat{C} \coloneqq \begin{pmatrix}
    1 & - \widehat{\gamma}_{1,2} & \cdots & -\widehat{\gamma}_{1,p} \\
    - \widehat{\gamma}_{2,1} & 1 & \cdots & - \widehat{\gamma}_{2,p}\\
    \vdots & \vdots & \ddots & \vdots\\
    - \widehat{\gamma}_{p,1} & - \widehat{\gamma}_{p,2} & \cdots & 1
\end{pmatrix}.
\end{align*}
Let $\sign\left(\widehat{\bm{\gamma}}_j\right) = \left(\sign\left(\widehat{\gamma}_{j, 1}\right),\cdots, \sign\left(\widehat{\gamma}_{j, p}\right)\right)^\top$. 
Then, we define
\begin{align*}
    \widehat{T}^{-2} \coloneqq \mathrm{diag}\left(\widehat{\tau}^2_1,\dots, \widehat{\tau}^2_p\right)
\end{align*}
where for each $j\in[p]$ and $d\in\{1, 0\}$,
\begin{align*}
    \widehat{\tau}^2_j \coloneqq \left( \mathbb{W}_{j} - \bm{W}_{(-j)}\widehat{\bm{\gamma}}_j \right)^2 / n   + \lambda_j \left\|\widehat{\bm{\gamma}}_j\right\|_1. 
\end{align*}
By using $\widehat{C}$ and $\widehat{T}^{-2}$, we construct the approximate inverses $\widehat{\Theta}$ by 
\[\widehat{\Theta}_{\mathrm{nodewise}} \coloneqq \widehat{T}^{-2}\widehat{C}. \]
Finally, we obtain the WTDL estimator with $\widehat{\Theta}_{\mathrm{nodewise}}$ as
\begin{align}
\label{decate_nodewise2}
\widehat{\bm{b}}^{\mathrm{WTDL}}_{n, \mathrm{nodewise}} \coloneqq \widehat{\bm{\beta}}^{\mathrm{WDML}}_n+ \widehat{\Theta}_{\mathrm{nodewise}} \bm{W}^\top \widehat{\bm{\varepsilon}}^{\mathrm{WDML}} / n.
\end{align}

We refer to this estimator as the WTDL estimator since we apply the DML and debiased Lasso together in weighted least squares. Here, it holds that 
\begin{align*}
\sqrt{n}\left(\widehat{\bm{b}}^{\mathrm{WTDL}}_{n, \mathrm{nodewise}} - \bm{\beta}_0\right) = \widehat{\Theta} \bm{W}^\top \widehat{\bm{\varepsilon}}^{\mathrm{WDML}} / \sqrt{n} - \widehat{\Delta}^{\mathrm{WDML}} + o_P(1),
\end{align*}
The detailed statement of this results are shown in Lemma~\ref{lem:conv_Q} with its proof. 
The term $\widehat{\Theta} \lambda \widehat{\bm{\kappa}} $ debiases the orignal Lasso estimator. The term comes from the KKT conditions. This term is similar to one in the debiased Lasso in \citet{vandeGeer2014} but different from it because we penalizes the difference $\bm{\beta}$.
As shown in Section~\ref{sec:theoretical}, we can develop the confidence intervals for the debiased CATE Lasso estimator if the inverses $\widehat{\Theta}$ is appropriately approximated.

\subsection{Summary: WTDL estimator with the nodewise Lasso Regression}
Lastly, we summarize the above arguments. First, in Section~\ref{sec:ind_common}, we assume sparsity directly to CATEs. In Section~\ref{sec:weight_tdl}, we discuss our estimation strategy. In Section~\ref{sec:diff_approx}, we define the DML-CATELasso estimator. Then, by using the debiased Lasso as in Section~\ref{sec:wtdl_reg}, we obtain the confidence interval about $\bm{\beta}_0$ by using an approximated inverse constructed using the node-wise regression. The pseudo-code is shown in Algorithm~\ref{alg}. Our main estimator is the WTDL estimator. In the following sections, we investigate the theoretical properties of the WTDL estimator.   

\begin{algorithm}[tb]
   \caption{WTDL}
   \label{alg}
\begin{algorithmic}
   \STATE {\bfseries Parameter:} The number of sample splitting $m \in  \mathbb{N}$ such that $m \geq 2$ and the regularization paramter $\lambda > 0$.
   \STATE Split $\mathcal{D}$ into $m$-sub-datasets $\{\mathcal{D}^\ell\}^m_{\ell = 1}$. We also construct $\{\overline{\mathcal{D}}^\ell\}^m_{\ell = 1}$.
   \STATE Obtain the estimator $\widehat{\mu}_{\ell, n}$ and $\widehat{\pi}_{\ell, n}$ of $\mu_0$ and $\pi_0$ using each sub-dataset $\overline{\mathcal{D}}^{\ell}$.
   \STATE Construct $\{\widehat{Q}^{\mathrm{DML}}_i\}^n_{i=1}$ and $\widehat{\sigma}_{\bar{\varepsilon}}(x)$, as in \eqref{eq:dml_est}.
   \STATE Construct $\{\widehat{Q}^{\mathrm{WDML}}_i\}^n_{i=1}$.
   \STATE Regress $\bm{W}$ on $\widehat{\mathbb{Q}}^{\mathrm{WDML}}$ with the Lasso regularization, and obtain an estimator $\widehat{\bm{\beta}}^{\mathrm{WDML}}_n$ of $\bm{\beta}_0$.
   \STATE Debias $\widehat{\bm{\beta}}^{\mathrm{WDML}}_n$ by using the debiased Lasso and nodewise regression, and obtain the WTDL estimator $\widehat{\bm{b}}^{\mathrm{WTDL}}_n$, as in \eqref{eq:wtdl_cate_Lasso}.
\end{algorithmic}
\end{algorithm}

\section{Theoretical Properties}
\label{sec:theoretical}
This section provides several theoretical properties of our WDML-CATELasso and WTDL estimators. 

\subsection{Consistency of the WDML-CATELasso Estimator}
\label{sec:res_cons}
First, we provides the consistency of the WDML estimator $\widehat{\bm{\beta}}^{\mathrm{WDML}}_n$. 
In our analysis, we a $\beta_{0, j}$
Our analysis starts from the introduction of the compatibility condition \citep{Buhlmann2011}, which plays an important role in identifiability of $\bm{\beta}_0$.
Let us denote the $j$-th element of $\bm{\beta}_0$. 
For the index set $\mathcal{S}_0 \subset [p]$ by $\beta_{0, j}$. Let us define $\bm{\beta}_{\mathcal{S}_0}$ and $\bm{\beta}_{\mathcal{S}^c_0}$ as vectors whose $j$-th elements are 
\begin{align*}
    \beta_{j, \mathcal{S}_0} \coloneqq \beta_{j}\mathbbm{1}[j\in \mathcal{S}_0],\qquad \beta_{j, \mathcal{S}^c_0} \coloneqq \beta_{j}\mathbbm{1}[j\notin \mathcal{S}_0],
\end{align*}
respectively. Under this definition, we have $\bm{\beta} = \bm{\beta}_{\mathcal{S}_0} + \bm{\beta}_{\mathcal{S}^c_0}$. Then, we define the compatibility condition. 
\begin{definition}[Compatibility condition. From (6.4) in \citet{Buhlmann2011}]
    We say that the compatibility condition holds for the set $S_0$, if there is a positive constant $\phi_0 > 0$ such that for all $\bm{\beta}$ satisfying $\|\bm{\beta}_{\mathcal{S}^c_0}\|_1 \leq 3 \|\bm{\beta}_{\mathcal{S}_0}\|_1$, we have
    \begin{align*}
        \|\bm{\beta}_{\mathcal{S}_0}\|^2_1 \leq s_0 \bm{\beta}^\top \widehat{\Sigma} \bm{\beta} / \phi^2_0. 
    \end{align*}
    We refer to $\phi^2_0$ as the compatibility constant.
\end{definition}

We denote the diagonal elements of a ($p\times p$)-matrix $A$ by $A_j$ for $j\in[p]$. 
Based on the compatibility condition, we assume the following. 
\begin{assumption}[From (A1) in \citet{vandeGeer2014}]
\label{asm:a1_vandegeer}
    The covariance $\widehat{\Sigma}$ satisfies the compatibility condition with compatibility constant $\phi^2_0 > 0$. Furthermore, $\max_{j, d} \widehat{\Sigma}^2_j \leq M^2$ holds for some $0 < M < \infty$. 
\end{assumption}

Let $\Sigma$ be the population of $\widehat{\Sigma}$. Then, in the following theorem, we show the consistency of the WDML-CATELasso estimator $\widehat{\bm{\beta}}^{\mathrm{WDML}}$ in a case where the covariates $\bm{W}_i$, treatment assignments $D_i$ are random and $p = p_n \to \infty$ as $n\to \infty$, and $\widetilde{W}$ is sub-Gaussian.
\begin{theorem}
\label{thm:nongaussian}
Assume that Assumptions~\ref{asmp:unconfounded}--\ref{asmp:coherent} hold. 
Suppose that a linear model in \eqref{eq:linear} holds with Assumptions~\ref{asm:bounded_output}--\ref{asmp:coherent}, \ref{asm:eror}, and \ref{asm:a1_vandegeer}. If $\bm{W}$ is sub-Gaussian and ${\Sigma}$ has a strictly positive smallest eigenvalue $\Lambda^2_{\min}$, satisfying $\max \Sigma_{j,j} = O(1)$  and $1/\Lambda^2_{\min} = O(1)$ as $n\to \infty$. Consider the WDML-CATELasso estimator $\widehat{\bm{\beta}}^{\mathrm{WDML}}$ and $\lambda \asymp \sqrt{\log(p) / n}$. If $s_0 = o(\sqrt{\log(p)/n})$ holds and either $\left\|\widehat{\mu}_\ell(d) - \mu_0(d)\right\|_2 = o_P(1)$ or $ \left\|\widehat{\pi}_\ell - \pi_0\right\|_2 = o_P(1)$ holds as $n\to \infty$ for all $\ell \in [m]$, then we have the following as $n \to \infty$:
        \begin{align*}
            &\left\| \widehat{\bm{\beta}}^{\mathrm{WDML}}_n- \bm{\beta}_0 \right\|_1 = o_P\left(s_0\sqrt{\log(p) / n}\right),\mbox{~~and}\\
            &\left\|\bm{W}\left(\widehat{\bm{\beta}}^{\mathrm{WDML}}_n- \bm{\beta}_0\right) \right\|^2_2 / n = o_P\left(s_0\log(p) / n\right).
        \end{align*}
%    hold as $n\to \infty$. 
\end{theorem}
This result implies consistency of $\widehat{\bm{\beta}}$; that is, $\widehat{\bm{\beta}}_n\xrightarrow{\mathrm{p}}\bm{\beta}_0$. The proof is provided in Appendix~\ref{sec:theoretical_fix}.

\subsection{Asymptotic Normality}
This section provides the asymptotic normality of the TDL estimator $\widehat{\bm{b}}^{\mathrm{TDL}}_n$. Before showing the result, we make the following assumption.
\begin{assumption}
\label{asm:conv_rate}
    For each $\ell \in [m]$ and each $d\in\{1, 0\}$, the following holds:
\begin{align*}
    &\left\|\widehat{\mu}_\ell(d) - \mu_0(d)\right\|_2 = o_P(1),\qquad \left\|\widehat{\pi}_\ell - \pi_0\right\|_2 = o_P(1),\\
    &\left\|\widehat{\mu}_\ell(d) - \mu_0(d)\right\|_2\left\|\widehat{\pi}_\ell - \pi_0\right\|_2 = o_P\left(1/n^{1/2}\right)),
\end{align*}
as $n \to \infty$.
\end{assumption}
\begin{assumption}[From (B1) in \citet{vandeGeer2014}]
\label{asm:b1_vandegeer}
    The variable $\bm{W}$ has either
    \begin{itemize}
        \item i.i.d. sub-Gaussian rows, i.e., $\max_{i\in\mathcal{N}}\sup_{\|v\|_2 \leq 1} \mathbb{E}_{P_0}\left[\exp\left(\left|\sum_{j\in[p]} v_j W_{i, j}\right|^2 / L^2\right)\right] = O(1)$ for some fixed constant $L > 0$.
        \item i.i.d. rows and for some $K \geq 1$, $\|\bm{W}\|_\infty = \max_{i, j}|W_{i, j}| = O(K)$ as $n\to\infty$.
    \end{itemize}
    The latter is referred to as the bounded case. The strongly bounded case assumes in addition that $\max_{j\in[p]}\|\widetilde{\bm{W}}_{(-j)}\widehat{\bm{\gamma}}_j\|_\infty = O(K)$. 
\end{assumption}
\begin{assumption}[From (B2) in \citet{vandeGeer2014}]
\label{asm:b2_vandegeer}
One of the followings holds:
        \begin{itemize}
        \item in the sub-Gaussian case, it holds that $\max_{j\in[p]} \sqrt{s_j \log(p) / n} = o(1)$ as $n\to\infty$.
        \item in the (strongly) bounded case, we assume that $\max_{j\in[p]}K^2s_j\sqrt{\log(p) / n} = o(1)$ as $n\to\infty$. 
    \end{itemize}
\end{assumption}
\begin{assumption}[From (B3) in \citet{vandeGeer2014}]
\label{asm:b3_vandegeer}
    The smallest eigenvalue $\Lambda^2_{\min}$ of $\Sigma$ is strictly positive and $1/\Lambda^2_{\min} = O(1)$ as $n\to\infty$. Moreover, $\max_{j}\Sigma_{j, j} = O(1)$ as $n\to\infty$.
\end{assumption}
\begin{assumption}[From (B4) in \citet{vandeGeer2014}]
\label{asm:b4_vandegeer}
    In the bounded case, it holds that $\max_{j}\mathbb{E}_{P_0}\left[\eta^4_{j}\right] = O(K^4)$ as $n\to\infty$. 
\end{assumption}

Then, we show the following lemmas. The first lemma is about the fast convergence of the asymptotic bias of the DML-CATELasso estimator. The proof of this lemma is shown in Appendix~\ref{appdx:lem:conv_Q}.
\begin{lemma}
    \label{lem:conv_Q}
    Suppose that Assumption~\ref{asm:conv_rate} holds. Then, we have
    \begin{align*}
        \sqrt{n}\left\|\widehat{\bm{v}} - \overline{\bm{v}}\right\|_{1} = o_P(1).
    \end{align*}
\end{lemma}

The second lemma is about the form of the TDL estimator. The proof is shown in Appendix~\ref{appdx:nongaussian}.
\begin{lemma}
    \label{lem:asymp_normal}
Suppose that Assumptions~\ref{asmp:unconfounded}--\ref{asmp:coherent} and \ref{asm:b1_vandegeer}--\ref{asm:b4_vandegeer} hold. 
Assume a linear model in \eqref{eq:linear} with Assumption~\ref{asm:eror}. If $\bm{W}$ is sub-Gaussian and the errors satisfy $\mathbb{E}_{P_0}\left[\exp(|\varepsilon| / L)\right] = O(1)$ for some fixed $L$ (sub-exponential). Consider the WTDL estimator with nodewise regression (\eqref{decate_nodewise2}) with $\lambda \asymp K_0 \sqrt{\log(p) / n}$. Assume that $K_0 s_0\log(p) / \sqrt{n} = o(1)$ and $\max_j K_0 s^d_j \sqrt{\log(p) / n} = o(1)$ for $d\in\{1, 0\}$ as $n\to\infty$. Then, we have
    \begin{align*}
    \sqrt{n}\left(\widehat{\bm{b}}^{\mathrm{WTDL}}_{n, \mathrm{nodewise}} - \bm{\beta}_0\right) = \widehat{\Theta} \bm{W}^\top \overline{\bm{\varepsilon}}^{\mathrm{Weight}} / \sqrt{n} - \widehat{\Delta}^{\mathrm{WDML}} + R,
    \end{align*}
    where $R \coloneqq \sqrt{n}\widehat{\Theta} \bm{W}^\top\Big(\widehat{\bm{v}} - \overline{\bm{v}}\Big)$, 
    and recall that $\widehat{\Delta}^{\mathrm{WDML}} = \sqrt{n}\widehat{\Theta}\widehat{\Sigma} \left( \widehat{\bm{\beta}}^{\mathrm{DML}}_n - \bm{\beta}_0\right)$. 
    Additionally, $\|R\|_1 = o_P(1)$ and $\|\Delta^{\mathrm{DML}}\|_1 = o_P(1)$  hold as $n\to \infty$. 
\end{lemma}
This result implies that under the raw-sparsity condition of $\bm{X}$, the bias term in the debiased CATE Lasso estimator converges to zero with the $1/\sqrt{n}$ order. Therefore, we can develop the asymptotic distribution with ignoring the bias term.

Then, from the central limit theorem, the confidence interval is given as follows.
\begin{theorem}[Asymptotic normality of the WTDL estimator]
Suppose that $\bm{W}$ is random, and the conditions in Lemma~\ref{lem:asymp_normal} hold. 
    Then, for each $j\in[p]$, the following holds as $n\to\infty$:
    \begin{align*}
       \frac{\sqrt{n}\left(\widehat{\bm{b}}^{\mathrm{WTDL}}_{n, \mathrm{nodewise}, j} - \bm{\beta}_{0, j}\right)}{\sqrt{\left(\widehat{\Theta}^\top_{\mathrm{nodewise}}  \widehat{\Sigma} \widehat{\Theta}_{\mathrm{nodewise}}\right)_{j, j}}} = G_j + o_P(1).
    \end{align*}
    where $G_j$ converges weakly to a standard Gaussian distribution as $n\to\infty$.
\end{theorem}
\begin{proof}
    Proposition~\ref{lem:random} gives us
\[\Big(\widehat{\Theta}^\top_{\mathrm{nodewise}}  \widehat{\Sigma} \widehat{\Theta}_{\mathrm{nodewise}}\Big)_{j, j}\xrightarrow{\mathrm{p}} \Sigma_{j,j}\ \ \ \mathrm{as}\ \ n\to\infty.\]
Then, 
\[\sqrt{n}\left(\widehat{\bm{b}}^{\mathrm{WTDL}}_{n, \mathrm{nodewise}, j} - \bm{\beta}_{0, j}\right) \xrightarrow{\mathrm{d}} \mathcal{N}\big(0, \Sigma_{j,j}\big)\qquad \mathrm{as}\ n\to\infty\]
holds. 
\end{proof}
This theorem gives us an asymptotic confidence interval for the WTDL estimator. Thus, we show the consistency and confidence intervals for the WTDL estimator.

\subsection{On the Convergence Rate Condition}
Assumption~\ref{asm:conv_rate} is difficult to guarantee when the dimension $p$ is large. For instance, it is common to deal with this by imposing sparsity or, in the case of using neural networks, by making assumptions about the structure of the neural network.

However, there are methods to ensure this condition without imposing sparsity. For example, one method is to use the benign-overfitting theory by \citet{Bartlett2020}. This approach demonstrates that appropriate convergence rates can be obtained even when $p$ is infinite, by using assumptions about the eigen space of the inverse matrix of the covariate 
$X$'s design matrix. Although confidence intervals are not provided, \citet{kato2022benignoverfitting} has discussed consistent estimation of CATE through linear regression in a setting where $p$ is infinite and sparsity is not imposed.

\section{Related Work}
\label{sec:related}
We denote $\widehat{\mu}(d)(x)$ and $\widehat{\pi}(d\mid x)$ as some estimators of $\mu_0(d)(x)$ and $\pi(d\mid x)$, respectively, to introduce related work. These estimators can vary depending on the methodology used. For instance, traditional approaches might employ parametric regression or kernel-based nonparametric regression to construct these estimators. In contrast, modern methodologies might utilize random forests or neural networks. The construction method impacts the confidence interval, especially due to the violation of the Donsker condition.

Pioneering studies on CATE estimation include those by \citet{Heckman1997} and \citet{Heckman2005}. Works by \citet{LeeWhang2009} and \citet{Hsu2017} discuss both estimation and hypothesis testing of CATEs. Studies by \citet{Cai2017,Cai2021} focus on confidence intervals in high-dimensional contexts. \citet{Abrevaya2015} addresses the nonparametric identification of CATEs and introduces a Nadaraya-Watson-based estimator.

The advent of machine learning algorithms has led to various proposed methods for estimating CATEs. \citet{Kunzel2019} summarize some of these methods as meta-learners, as shown in the subsequent section.

Furthermore, the literature on high-dimensional linear regression deserves mention. Beyond the Lasso estimator, numerous high-dimensional linear regression techniques have been proposed, including Ridge \citep{Hoerl1970} and Elastic Net \citep{ZouHastie2005} under sparsity conditions. High-dimensional regression methods not involving regularization, known as ridgeless estimation or interpolating estimators, have also been developed \citep{Bartlett2020}. \citet{Bartlett2020} introduces the benign-overfitting framework for interpolating estimators, and \citet{Tsigler2023} demonstrates benign overfitting in ridge regression. The debiased Lasso was introduced by \citet{zhangzhang2014}, with extensions by \citet{Javanmard14a} and \citet{vandeGeer2014}. Related studies include those by \citet{Belloni2014}, \citet{Belloni2014Post}, and \citet{Belloni2016}, with \citet{Belloni2014} considering treatment effect estimation in a similar vein to our work. 

In the literature of semiparametric analysis, we can consider an estimator including nuisance parameters estimated by some machine learning algorithms. If the nuisance estimators do not satisfy the Donsker condition, the estimator of the parameter of interest can be potentially biased. To debias the estimator,  \citet{ChernozhukovVictor2018Dmlf} proposes DML, which is a refinement of the sample-splitting methods used in early works such as \citet{klaassen1987} and \citet{ZhengWenjing2011CTME}. \citet{Fan2022} proposes Lasso with DML and multiplier-bootstrap to obtain a confidence interval. See also \citet{Chernozhukov2024book}.

Additionally, numerous methods have also been proposed for CATE estimation \citep{Li2017,Kallus2017,Powers2017,Subbaswamy2018,Zhao2019,Nie2020,Hahn2020,sugasawa2023bayesian,lan2023causal,kato2023cate,Jin2023}.

%However, the IPW-learner requires the true value for the propensity score $p(D = 1 \mid X)$. When it is unknown and replace it with an estimate, we cannot enjoy the unbiasedness. Furthermore, under high-dimensional models, the estimation of $p(D = 1 \mid X)$ itself becomes problematic because we need to assume something like sparsity for estimating $p(D = 1 \mid X)$ with high-dimensional $X$. The DR-learner further requires an estimate of ${\mu}_0(0)(x)$ to estimate $f_0(X)$, in addition to $p(D = 1 \mid X)$. In contrast, our method does not suffer from the problem.

\section{Conclusion}
This study delved into th Lasso estimation of CATEs with sparsity, we introduced a three-stage estimation methodology: (i) estimating nuisance parameters, which include conditional expected outcomes and propensity scores; (ii) estimating the difference in outcomes, denoted as $\widehat{\mathbb{Q}}^{\mathrm{DML}}$, through cross-fitting and regressing covariates $\bm{W}$ on this difference using Lasso regularization; and (iii) applying the debiased Lasso technique to mitigate the bias introduced by the Lasso regularization. The estimator derived from this approach is termed the TDL estimator. We demonstrated the consistency and asymptotic normality of the TDL estimator and discuss the conditions under which asymptotic normality is achieved.

\bibliographystyle{tmlr}
\bibliography{arXiv.bbl}

\onecolumn
\appendix

\section{Details of the Estimation Strategy}
\label{sec:det_est_strategy}
This section explains how we derive our proposed TDL estimator with an intuitive explanation. 

\subsection{Lasso with Oracle Nuisance Parameters}
\label{sec:oracledrcate}
First, we start our explanation from the estimation of CATEs using the true nuisance parameters $\mu_0(d)(X) = \mathbb{E}_{P_0}[Y(d) \mid X]$ and $\pi_0(d\mid X) = \mathbb{E}_{P_0}[D = d \mid X]$ for each $d\in\{1, 0\}$. An estimator propsosed in this section is infeasible in practice but gives us important intuitions for constructing a feasible estimator in practice. 

Let us define 
\[ \overline{Q}_i \coloneqq Q\big(Y_i, D_i, X_i; \mu_0, \pi_0\big)\ \ \mathrm{and}\ \ \overline{\mathbb{Q}} \coloneqq \Big(\overline{Q}_i\ \overline{Q}_2\ \cdots \overline{Q}_n\Big)^\top.\] 
Then, we define the Oracle DR CATE Lasso (Oracle-DR-CATELasso) estimator $\overline{\bm{\beta}}_n$  as 
\begin{align}
\label{eq:Oracle-DRCATELasso}
&\overline{\bm{\beta}}_n \in \argmin_{\bm{\beta}\in\mathbb{R}^p} \left\{\left\|\overline{\mathbb{Q}} - \bm{X}\bm{\beta} \right\|^2_2 / n+ \lambda\left\|\bm{\beta}\right\|_1\right\}.
\end{align}

We can show the consistency by using the standard approaches in Lasso, as in \citet{vandeGeer2008}. However, the Oracle-DR-CATELasso estimator does not have the asymptotic normality. Therefore, we debiase the bias caused by Lasso using the debiased Lasso technique, such as \citet{vandeGeer2014}.

The Oracle-DR-CATELasso is infeasible in practice since $\mu_0$ and $\pi_0$ are unknown, which are replaced with some estimators. Furthermore, the conditional variance $\sigma^2_{\bar{\varepsilon}}(x)$ of the error term in the linear regression model \eqref{eq:cate_linear_model} varies across  $x\in\mathcal{X}$, which affects the efficiency of the estimator In construction of an estimator used in practice, we deal with these problems.

\subsection{DR CATE Lasso with Estimated Nuisance Parameters}
\label{sec:consistentdrcate}
In Section~\ref{sec:oracledrcate}. we defined Lass using the oracle nuisance parameters $\mu_P(d)(X)$ and $\pi_P(X)$. Since they are usually unknown, we consider replacing them with their estimators. In fact, if we have any consistent estimators of $\mu_P(d)(X)$ and $\pi_P(X)$, we can obtain consistent estimators $\bm{\beta}$ and $\tau(x)$.

\paragraph{Construction of a score function.} Let $\widehat{\mu}$ and $\widehat{\pi}$ be some estimators of $\mu_P(d)(X)$ and $\pi_P(X)$ under $P_0$. 
By replacing $\mu_0$ and $\pi_0$ in $\mathbb{Q}$ with the estimators $\widehat{\mu}$ and $\widehat{\pi}$, we define 
\[ \widehat{Q}_i \coloneqq Q\big(Y_i, D_i, X_i; \widehat{\mu}, \widehat{\pi}\big)\ \ \mathrm{and}\ \ \widehat{\mathbb{Q}} \coloneqq \Big(\widehat{Q}_i\ \widehat{Q}_2\ \cdots \widehat{Q}_n\Big)^\top.\]

Here, the following linear regression model holds under $P_0$:
\begin{align*}
   \widehat{\mathbb{Q}} = \bm{X}\bm{\beta} + \widehat{\bm{\varepsilon}},
\end{align*}
where $\widehat{\bm{\varepsilon}}$ is an unobservable $n$-dimensional vector whose $i$-th element $\widehat{{\varepsilon}}_i$ is equal to 
\begin{align*}
    \widehat{{\varepsilon}}_i \coloneqq \frac{\mathbbm{1}\big[D_i = 1\big]\big(Y_i - \widehat{\mu}(1)(X_i)\big)}{\widehat{\pi}(1\mid X_i)} - \frac{\mathbbm{1}\big[D_i = 0\big]\big(Y_i - \widehat{\mu}(0)(X_i)\big)}{\widehat{\pi}(0\mid X_i)} - \Big\{Y_i(1) - \widehat{\mu}(1)(X_i)\Big\} + \Big\{Y_i(0) - \widehat{\mu}(0)(X_i)\Big\} + \varepsilon_i(1) - \varepsilon_i(0).
\end{align*}
Note that without additional assumptions, $\mathbb{E}_{P_0}\left[\widehat{\varepsilon}_i \mid \mathbbm{X}_i\right] = 0$ does not holds in general due to the estimation error.

\paragraph{DR CATE Lasso.}
Then, we define the DR CATE Lasso (DR-CATELasso) $\widehat{\bm{\beta}}_n$ estimator as 
\begin{align*}
&\widehat{\bm{\beta}}_n\in \argmin_{\bm{\beta}\in\mathbb{R}^p} \left\{\left\|\overline{\mathbb{Q}} - \bm{X}\bm{\beta} \right\|^2_2 / n+ \lambda\left\|\bm{\beta}\right\|_1\right\},
\end{align*}
where $\lambda > 0$ is a penalty coefficient. This estimator is doubly robust in the sense that either $\widehat{\mu}$ or $\widehat{\pi}$ is consistent, $\widehat{\bm{\beta}}_n$ is also consistent. We show the proof in Lemma~\ref{lem:nongaussian}.

\subsection{Derivation of the WTDL Estimator}
Based on the arguments in Appendix~\ref{sec:det_est_strategy}, we design our proposed  WTDL estimator.

\paragraph{Two sources of the bias in the DR-CATELasso estimator.} In the DR-CATELasso estimator, there are two sources of the bias. First, even without using estimators of $\mu_0$ and $\pi_0$, the Oracle-DR-CATELasso estimator is biased due to the use of the Lasso regularization. Furthermore, when we estimate $\mu_0$ and $\pi_0$ using estimators that do not satisfy the Donsker condition, we cannot obtain the bias due to the use of such estimators remain in $\sqrt{n}$ approximation. We debias these biases by using the techniques of the debiased Lasso and DML, respectively. First, we debias the bias caused from the use of non-Donsker estimators for $\mu_0$ and $\pi_0$ using the DML. Then, we debias the bias caused from the Lasso regularization by using the debiased Lasso.

\paragraph{(A) DML for the CATE estimator.} To debias the bias caused in estimation of $\mu_0$ and $\pi_0$, we employ the cross-fitting in DML. The cross-fitting is a variant of the sample-splitting technique that has been utilized in semiparametric analysis, such as ..., and refined by \citet{ChernozhukovVictor2018Dmlf}. In this technique, we first split samples into $m$ subgroups with sample sizes $n_\ell = n / m$. For simplicity, without loss of generality, we assume that $n/m$ is an integer. For each $\ell \in [m]$, we define $\mathcal{D}^\ell \coloneqq \left\{(Y_{\ell, (i)}, D_{\ell, (i)}, X_{\ell, (i)})\right\}^{n_\ell}_{i=1}$ and $\overline{\mathcal{D}}^\ell \coloneqq \mathcal{D} \backslash \mathcal{D}^\ell$, where $(Y_{\ell, (i)}, D_{\ell, (i)}, X_{\ell, (i)})$ denotes an $i$-th element in the $\ell$-th subgroup. Let $\widehat{\mu}_{\ell, n}$ and $\widehat{\pi}_{\ell, n}$ be some estimators constructed only using $\overline{\mathcal{D}}^\ell$, without using $\mathcal{D}^\ell$. The estimators are assumed to satisfy a convergence rate condition described in Assumption~\ref{asm:conv_rate}. Then, we construct a score function as
\[ \widehat{Q}^{\mathrm{DML}}_i \coloneqq Q\big(Y_{\ell(i), i}, D_{\ell(i), i}, X_{\ell(i), i}; \widehat{\mu}_{\ell(i), n}, \widehat{\pi}_{\ell(i), n}\big)\ \ \mathrm{and}\ \ \widehat{\mathbb{Q}}^{\mathrm{DML}} \coloneqq \Big(\widehat{Q}^{\mathrm{DML}}_1\ \widehat{Q}^{\mathrm{DML}}_2\ \cdots \widehat{Q}^{\mathrm{DML}}_n\Big)^\top,\]
where $\ell(i)$ denotes a subgroup $\ell \in [m]$ such that sample $i$ belongs to it. 
Here, the following linear regression model holds under $P_0$:
\begin{align*}
   \widehat{\mathbb{Q}}^{\mathrm{DML}} = \bm{X}\bm{\beta}_0 + \widehat{\bm{\varepsilon}}^{\mathrm{DML}},
\end{align*}
where $\widehat{\bm{\varepsilon}}^{\mathrm{DML}}$ is an unobservable $n$-dimensional vector whose $i$-th element $\widehat{{\varepsilon}}^{\mathrm{DML}}_i$ is equal to 
\begin{align*}
     \widehat{{\varepsilon}}^{\mathrm{DML}}_i &\coloneqq \widehat{u}^{\mathrm{DML}}_i + \varepsilon_i(1) - \varepsilon_i(0),\\
     \widehat{u}^{\mathrm{DML}}_i & \coloneqq \frac{\mathbbm{1}\big[D_i = 1\big]\big(Y_i - \widehat{\mu}_{\ell(i), n}(1)(X_i)\big)}{\widehat{\pi}_{\ell(i), n}(1\mid X_i)} - \frac{\mathbbm{1}\big[D_i = 0\big]\big(Y_i - \widehat{\mu}_{\ell(i), n}(0)(X_i)\big)}{\widehat{\pi}_{\ell(i), n}(0\mid X_i)}\\
     &\ \ \ \ \ \ \ \ \  - \Big\{Y_i(1) - \widehat{\mu}_{\ell(i), n}(1)(X_i)\Big\}+ \Big\{Y_i(0) - \widehat{\mu}_{\ell(i), n}(0)(X_i)\Big\} 
\end{align*}
In Lemma~\ref{lem:conv_Q}, we show that for each $i\in[n]$, $\sqrt{n}| \widehat{u}^{\mathrm{DML}}_i - \overline{u}_i \mid = o_P(1)$ holds, which plays an important role for deriving a confidence interval. We refer to this estimator as the DML-CATELasso estimator. Note that this estimator does not have the asymptotic normality due to the bias caused from the use of the Lasso regularization. Furthermore, if we are only interested in the consistency, we do not have to apply the cross-fitting.

Then, we estimate the Lasso estimator as
\begin{align}
\label{eq:DML-CATELasso}
&\widehat{\bm{\beta}}^{\mathrm{DML}}_n \in \argmin_{\bm{\beta}\in\mathbb{R}^p} \left\{\left\|\widehat{\mathbb{Q}}^{\mathrm{DML}} - \bm{X}\bm{\beta} \right\|^2_2 / n+ \lambda\left\|\bm{\beta}\right\|_1\right\}.
\end{align}

As a result of optimization, the solution $\widehat{\bm{\beta}}^{\mathrm{DML}}_n$ of Lasso regression satisfies the following KKT conditions:
\begin{align*}
&\bm{W}^\top\left(\widehat{\mathbb{Q}}^{\mathrm{DML}} - \bm{W}\widehat{\bm{\beta}}^{\mathrm{DML}}_n\right) / n = \lambda \widehat{\bm{\kappa}}^{\mathrm{DML}},
\end{align*}
where $\widehat{\bm{\kappa}}^{\mathrm{DML}} \coloneqq \sign\left(\widehat{\bm{\beta}}^{\mathrm{DML}}_n\right) = \Big(\sign\left(\widehat{\beta}^{\mathrm{DML}}_{n, 1}\right),\cdots, \sign\left(\widehat{\beta}^{\mathrm{DML}}_{n, p}\right)\Big)^\top$, and for each $j\in[p]$, 
\begin{align}
    \widehat{\kappa}^{\mathrm{DML}}_j = \sign\left(\widehat{\beta}^{\mathrm{DML}}_{n, j}\right) \in \begin{cases}
    1 & \mathrm{if}\quad \widehat{\beta}^{\mathrm{DML}}_{n, j} > 0\\
    [-1, 1] & \mathrm{if}\quad \widehat{\beta}^{\mathrm{DML}}_{n, j} = 0\\
    - 1 & \mathrm{if}\quad \widehat{\beta}^{\mathrm{DML}}_{n, j} < 0
    \end{cases}.
\end{align}

\paragraph{(B) Debiased Lasso for the DML-CATELasso estimator.} 
Then, we consider debiasing the bias caused from the Lasso regularization in the WDML-CATELasso estimator to obtain a confidence interval bout $\bm{\beta}_0$. We first focus on the Oracle-DR-CATELasso estimator $\overline{\bm{\beta}}_n$. Let $\widehat{\Sigma}^{\bm{X}} = \bm{X}^\top\bm{X} / n$. Then, the KKT condition in \eqref{eq:Oracle-DRCATELasso} yields 
\begin{align*}
&\widehat{\Sigma}^{\bm{X}} \left(\overline{\bm{\beta}}_n - \bm{\beta}_0\right)
+ \lambda \overline{\bm{\kappa}} =  \bm{X}^\top \overline{\bm{\varepsilon}} / n,
\end{align*}
where we used $\overline{\mathbb{Q}} = \bm{X}\bm{\beta}_0 + \overline{\bm{\varepsilon}}$, and 
\begin{align*}
&\bm{X}^\top\left(\overline{\mathbb{Q}}  - \bm{X}\overline{\bm{\beta}}_n\right)/n = \bm{X}^\top\left(\bm{X}\bm{\beta}_0 + \overline{\bm{\varepsilon}}- \bm{X}\overline{\bm{\beta}}_n\right) / n = \widehat{\Sigma}^{\bm{X}}\left(\bm{\beta}_0 - \overline{\bm{\beta}}_n\right) + \bm{X}^\top\overline{\bm{\varepsilon}}/ n.
\end{align*}
Suppose that $\widehat{\Theta}^{\bm{X}}$ is a reasonable approximation for the inverse of $\widehat{\Sigma}^{\bm{X}}$. Then, it holds that
\begin{align*}
\overline{\bm{\beta}}_n + \widehat{\Theta}^{\bm{X}} \lambda \overline{\bm{\kappa}} - \bm{\beta}_0 & = \widehat{\Theta}^{\bm{X}} \bm{X}^\top \overline{\bm{\varepsilon}} / n - \overline{\Delta} / \sqrt{n},
\end{align*}
where $\overline{\Delta} \coloneqq \sqrt{n}\widehat{\Theta}^{\bm{X}}\widehat{\Sigma}^{\bm{X}} \left( \overline{\bm{\beta}}_n - \bm{\beta}_0\right)$. 
Here, $\overline{\Delta} / \sqrt{n}$ is the bias term that is expected to approaches zero as $n \to \infty$. Then, the debiased Lasso defines $\overline{\bm{b}} \coloneqq \overline{\bm{\beta}}_n + \widehat{\Theta}^{\bm{X}} \lambda\overline{\bm{\kappa}}$ as an estimator, where $\sqrt{n}\left(\overline{\bm{b}} - \bm{\beta}_0\right) = \widehat{\Theta}^{\bm{X}} \bm{X}^\top \overline{\bm{\varepsilon}} / n - \overline{\Delta} / \sqrt{n}$ holds. Furthermore, $\widehat{\Theta}^{\bm{X}} \bm{X}^\top \overline{\bm{\varepsilon}} / \sqrt{n}$ converges to a Gaussian distribution from the central limit theorem and $\Delta$ converges to zero as $n\to \infty$. 

Lastly, we relates the debiased Oracle-DR-CATELasso estimator $\overline{\bm{b}}$ to the debiased Lasso estimator based on the DML-CATELasso estimator. Based on the above arguments, we define the debiased Lasso estimator based on the DML-CATELasso estimator as
\begin{align*}
\widehat{\bm{b}}^{\mathrm{TDL}}_n
&\coloneqq \widehat{\bm{\beta}}^{\mathrm{DML}}_n + \widehat{\Theta}^{\bm{X}} \lambda \widehat{\bm{\kappa}}^{\mathrm{DML}} = \widehat{\bm{\beta}}^{\mathrm{DML}}_n + \widehat{\Theta}^{\bm{X}} \bm{X}^\top\left(\widehat{\mathbb{Q}}^{\mathrm{DML}} - \bm{X}\widehat{\bm{\beta}}^{\mathrm{DML}}_n\right) / n.
\end{align*}
We refer to this estimator as the TDL estimator since we apply the DML and debiased Lasso together. Here, it holds that 
\begin{align}
\sqrt{n}\left(\widehat{\bm{b}}^{\mathrm{TDL}}_n - \bm{\beta}_0\right) = \widehat{\Theta}^{\bm{X}} \bm{X}^\top \widehat{\bm{\varepsilon}}^{\mathrm{DML}} / \sqrt{n} - \widehat{\Delta}^{\mathrm{DML}} + o_P(1),
\end{align}
where $\widehat{\Delta}^{\mathrm{DML}} \coloneqq \sqrt{n}\widehat{\Theta}^{\bm{X}}\widehat{\Sigma}^{\bm{X}} \left( \widehat{\bm{\beta}}^{\mathrm{DML}}_n - \bm{\beta}_0\right)$ is the bias term that is expected to approaches zero as $n \to \infty$.
The detailed statement of this results are shown in Lemma~\ref{lem:conv_Q} with its proof. 
The term $\widehat{\Theta} \lambda \widehat{\bm{\kappa}} $ debiases the orignal Lasso estimator. The term comes from the KKT conditions. This term is similar to one in the debiased Lasso in \citet{vandeGeer2014} but different from it because we penalizes the difference $\bm{\beta}$.
As shown in Section~\ref{sec:theoretical}, we can develop the confidence intervals for the debiased CATE Lasso estimator if the inverses $\widehat{\Theta}$ is appropriately approximated. 

\paragraph{(D) Nodewise Lasso regression for consructing $\widehat{\Theta}$.}
Next, we turn our attention to the construction of the approximate inverses $\widehat{\Theta}$. Estimation of $\widehat{\Theta}$ is also a core interest in this literature because which sparsity variant is needed to establish asymptotic normality highly depends on whether we know $\widehat{\Theta}$ or not. \citet{vandeGeer2014}, \citet{Javanmard14a}, \citet{Cai2017}, and \citet{Javanmard2018} have addressed this issue. Following \citet{vandeGeer2014}, we employ the nodewise Lasso regression \citep{Meinshausen2006}. We modify the nodewise Lasso regression proposed in \citet{vandeGeer2014} for CATE estimation. Let $\bm{X}_{(-j)}$ be the design submatrix excluding the $j$-th column. In our construction, for each $j\in[p]$, we carry out the nodewise Lasso regression as follows:
\begin{align*}
\widehat{\bm{\gamma}}^{\bm{X}}_j = \argmin_{\bm{\gamma}\in \mathbb{R}^{p-1}} &\Bigg\{\left\| \mathbb{X}_j - \bm{X}_{(-j)} \bm{\gamma} \right\|^2_2 /n + 2\lambda_j \|\bm{\gamma}\|_1 \Bigg\}, 
\end{align*}
with components of $\widehat{\bm{\gamma}}^{\bm{X}}_j = \{\widehat{\gamma}^{\bm{X}}_{j,k}; k \in [p], k\neq j\}$. Denote by 
\begin{align*}
\widehat{C}^{\bm{X}} \coloneqq \begin{pmatrix}
    1 & - \widehat{\gamma}^{\bm{X}}_{1,2} & \cdots & -\widehat{\gamma}^{\bm{X}}_{1,p} \\
    - \widehat{\gamma}^{\bm{X}}_{2,1} & 1 & \cdots & - \widehat{\gamma}^{\bm{X}}_{2,p}\\
    \vdots & \vdots & \ddots & \vdots\\
    - \widehat{\gamma}^{\bm{X}}_{p,1} & - \widehat{\gamma}^{\bm{X}}_{p,2} & \cdots & 1
\end{pmatrix}.
\end{align*}
Let $\sign\left(\widehat{\bm{\gamma}}^{\bm{X}}_j\right) = \left(\sign\left(\widehat{\gamma}^{\bm{X}}_{j, 1}\right),\cdots, \sign\left(\widehat{\gamma}^{\bm{X}}_{j, p}\right)\right)^\top$. 
Then, we define
\begin{align*}
    \left(\widehat{T}^{\bm{X}}\right)^{-2} \coloneqq \mathrm{diag}\left(\left(\widehat{\tau}^{\bm{X}}_1\right)^2,\dots, \left(\widehat{\tau}^{\bm{X}}_p\right)^2\right)
\end{align*}
where for each $j\in[p]$ and $d\in\{1, 0\}$,
\begin{align*}
    \left(\widehat{\tau}^{\bm{X}}_j\right) \coloneqq \left( \mathbb{X}_{j} - \bm{X}_{(-j)}\widehat{\bm{\gamma}}^{\bm{X}}_j \right)^2 / n   + \lambda_j \left\|\widehat{\bm{\gamma}}^{\bm{X}}_j\right\|_1. 
\end{align*}
By using $\widehat{C}^{\bm{X}}$ and $\left(\widehat{T}^{\bm{X}}\right)^{-2}$, we construct the approximate inverses $\widehat{\Theta}^{\bm{X}}$ by 
\[\widehat{\Theta}^{\bm{X}}_{\mathrm{nodewise}} \coloneqq \left(\widehat{T}^{\bm{X}}\right)^{-2}\widehat{C}^{\bm{X}}. \]
Then, we define the TDL estimator with the nodewise regression as 
\begin{align}
\sqrt{n}\left(\widehat{\bm{b}}^{\mathrm{TDL}}_{n, \mathrm{nodewise}} - \bm{\beta}_0\right) = \widehat{\Theta}^{\bm{X}}_{\mathrm{nodewise}} \bm{X}^\top \widehat{\bm{\varepsilon}}^{\mathrm{DML}} / \sqrt{n} - \Delta + o_P(1),
\end{align}

\paragraph{(E) Weighted least squares.} Under the standard arguments in the debiased Lasso, we expect that the TDL estimator $\sqrt{n}\left(\widehat{\bm{b}}^{\mathrm{TDL}}_{n, \mathrm{nodewise}} - \bm{\beta}_0\right)$ converges to the Gaussian distribution. However, since the variance $\sigma^2_{\bar{\varepsilon}}(x)$ of the error term depends on $x$, we can still minimize the asymptotic varaince of an estimator of $\bm{\beta}_0$ by using the weighted least squares. Then. we define the WTDL estimator as in Section~\ref{sec:cateLasso}.

\section{Proof of Theorem~\ref{thm:nongaussian}}
\label{sec:theoretical_fix}
This section provides the proof of Theorem~\ref{thm:nongaussian}. Instead of the proof of the consistency of $\widehat{\bm{\beta}}^{\mathrm{WDML}}_n$, we show the consistency of the DR-CATELasso estimator $\widehat{\bm{\beta}}_n$, which includes the WDML-CATELasso estimator $\widehat{\bm{\beta}}^{\mathrm{WDML}}_n$ as a special case. 

We first consider a case where $\bm{X}$ and $\mathbb{D}$ are non-random variables, called a fixed design. Under the fixed design, we develop an oracle inequality (non-asymptotic estimation error bound) that holds for the DR-CATELasso estimator $\widehat{\bm{\beta}}_n$. 
\begin{lemma}[Oracle inequality]
\label{thm:oracle}
Assume a linear model in \eqref{eq:linear} with fixed design for $\bm{X}$ and $\mathbb{D}$, which satisfies Assumptions~\ref{asm:bounded_output}--\ref{asmp:coherent}, \ref{asm:eror} and \ref{asm:a1_vandegeer}. Also suppose that $\varepsilon(d)$ follows a centered sub-Gaussian distribution  with variance $\sigma^2_{\bar{\varepsilon}}$. Let $t>0$ be arbitrary. Consider the DR-CATELasso estimator $\widehat{\bm{\beta}}^{\mathrm{DML}}_n$ with regularization parameter $\lambda \geq 3 M \sigma_{\bar{\varepsilon}}\sqrt{\frac{2\left(t^2 + \log(p)\right)}{n}}$. Then, with probability at least $1 - 2\exp\left(-t^2\right)$, 
        \begin{align*}
            &\left\| \widehat{\bm{\beta}}_n- \bm{\beta}_0 \right\|_1 \leq C_1 \lambda s_0 / \phi^2_0 + S / \lambda,\mbox{~~and}\\
            &\left\|\bm{W}\left(\widehat{\bm{\beta}}_n- \bm{\beta}_0\right) \right\|^2_2 / n \leq C_2 \lambda^2 s_0 / \phi^2_0 + S,
        \end{align*}
        hold, where $C_1, C_2 > 0$ are some universal constants, and $S \coloneqq 2\left\|\overline{\bm{\varepsilon}}^\top \bm{W}\left(\widehat{\bm{\beta}}_n- \bm{\beta}_0\right) \right\|_{\infty}/ n$.
\end{lemma}
The proof is provided in the subsequent section. 

Based on Lemma~\ref{thm:oracle}, we show the following result about the consistency of the DR-CATELasso estimator $\widehat{\bm{\beta}}_n$, which includes Theorem~\ref{sec:theoretical_fix} as a special case with a specific choice of $\widehat{\mu}(d)$ and $\widehat{\pi}$.
\begin{lemma}
\label{lem:nongaussian}
Suppose that Assumptions~\ref{asmp:unconfounded}--\ref{asmp:coherent} hold. 
Assume a linear model in \eqref{eq:linear} with Assumptions~\ref{asm:bounded_output}--\ref{asmp:coherent}, \ref{asm:eror}, and \ref{asm:a1_vandegeer}. If $\bm{W}$ is sub-Gaussian and ${\Sigma}$ has a strictly positive smallest eigenvalue $\Lambda^2_{\min}$, satisfying $\max \Sigma_{j,j} = O(1)$  and $1/\Lambda^2_{\min} = O(1)$ as $n\to \infty$. Consider the DR-CATELasso estimator $\widehat{\bm{\beta}}$ and $\lambda \asymp \sqrt{\log(p) / n}$. If $s_0 = o(\sqrt{\log(p)/n})$ holds and either $\left\|\widehat{\mu}(d) - \mu_0(d)\right\|_2 = o_P(1)$ or $ \left\|\widehat{\pi} - \pi_0\right\|_2 = o_P(1)$ holds as $n\to\infty$, then we have the following as $n \to \infty$:
        \begin{align*}
            &\left\| \widehat{\bm{\beta}}_n- \bm{\beta}_0 \right\|_1 = o_P\left(s_0\sqrt{\log(p) / n}\right),\mbox{~~and}\\
            &\left\|\bm{W}\left(\widehat{\bm{\beta}}_n- \bm{\beta}_0\right) \right\|^2_2 / n = o_P\left(s_0\log(p) / n\right).
        \end{align*}
%    hold as $n\to \infty$. 
\end{lemma}
\begin{proof}[Proof of Lemma~\ref{lem:nongaussian}]
First, we show that $S / \lambda = o_P(1)$ as $n\to\infty$. We have
\begin{align*}
    S = 2\left\|\overline{\bm{\varepsilon}}^\top \bm{W}\left(\widehat{\bm{\beta}}_n- \bm{\beta}_0\right) \right\|_{\infty}/ n\leq 2\left\|\overline{\bm{\varepsilon}}\right\|_{\infty}\left\| \bm{W}\left(\widehat{\bm{\beta}}_n- \bm{\beta}_0\right) \right\|_1/ n.
\end{align*}

If either $\widehat{\mu}$ or $\widehat{\pi}$ converges to the true function in probability, we have
\begin{align*}
    &\widehat{\varepsilon}_i - \overline{\varepsilon}_i\\
    &=\frac{\mathbbm{1}\big[D_i = 1\big]\big(Y_i - \widehat{\mu}(1)(X_i)\big)}{\widehat{\pi}(1\mid X_i)} - \frac{\mathbbm{1}\big[D_i = 0\big]\big(Y_i - \widehat{\mu}(0)(X_i)\big)}{\widehat{\pi}(0\mid X_i)} - \Big\{Y_i(1) - \widehat{\mu}(1)(X_i)\Big\} + \Big\{Y_i(0) - \widehat{\mu}(0)(X_i)\Big\}\\
    &\ \ \ + \varepsilon_i(1) - \varepsilon_i(0)\\
    &\ \ \ - \frac{\mathbbm{1}\big[D_i = 1\big]\big(Y_i - \mu(1)(X_i)\big)}{\pi(1\mid X_i)} + \frac{\mathbbm{1}\big[D_i = 0\big]\big(Y_i - \mu(0)(X_i)\big)}{\pi(0\mid X_i)} + \Big\{Y_i(1) - \mu(1)(X_i)\Big\} - \Big\{Y_i(0) - \mu(0)(X_i)\Big\}\\
    &\ \ \ - \varepsilon_i(1) + \varepsilon_i(0)\\
    &= o_P(1),
\end{align*}
as $n\to\infty$. Here, we also have 
\[\left\| \bm{W}\left(\widehat{\bm{\beta}}_n - \bm{\beta}_0\right) \right\|_1/ (\lambda n) \leq \sqrt{n}\left\| \bm{W}\left(\widehat{\bm{\beta}}_n - \bm{\beta}_0\right) \right\|_2/ (\lambda n) \leq \left(\sqrt{C_2} s_0 / \phi_0 + \sqrt{S} / \lambda\right) / \sqrt{n}.\]
Therefore, we have 
\[S / \lambda = o_P\left(\left\{\sqrt{C_2} s_0 / \phi_0 + \sqrt{S} / \lambda\right\} / \sqrt{n}\right),\]
which implies $S  = o_P\left(\left\{\sqrt{C_2}\lambda s_0 / \phi_0 + \sqrt{S} \right\} / \sqrt{n}\right) = o_P(1)$

Second, for $C_1 \lambda s_0 / \phi^2_0$, Theorem~1.6 in \citet{zhou2009restricted} (a sub-Gaussian extension of Theorem~1 in \citet{Raskutti2010}) implies that there exists a constant $L=O(1)$ as $n\to \infty$, where $L$ depends on $\Lambda_{\min}$ and the compatibility condition holds with compatibility constant $\phi^2_0 > 1/L^2$ with probability approaches one. Combining this result and Theorem~\ref{thm:oracle} yields the statement.

\end{proof}

\subsection{Basic Inequality}
Consider a fixed design and show Lemma~\ref{thm:oracle}.
Following Lemma~6.1 in \citet{Buhlmann2011}, our analysis starts from the derivation of an oracle inequality, which plays an important role in subsequent analysis.
\begin{lemma}[Basic inequality. Corresponding to Lemma~6.1 in \citet{Buhlmann2011}] \label{lem:basic_ineq}
    \begin{align*}
        &\left\| \bm{W}\left(\widehat{\bm{\beta}}_n- \bm{\beta}_0\right) \right\|^2_2 / n  + \lambda  \left\| \widehat{\bm{\beta}}_n\right\|_1\leq 2\widehat{\bm{\varepsilon}} \bm{W}\left(\widehat{\bm{\beta}}_n- \bm{\beta}_0\right) / n  + \lambda \|\bm{\beta}_0\|_1.
    \end{align*}
\end{lemma}
\begin{proof}
Recall that we estimate $\widehat{\bm{\beta}}$ as
\begin{align*}
\widehat{\bm{\beta}}_n&\coloneqq \argmin_{\bm{\beta}\in\mathbb{R}^p} \left\{
\left\|\widehat{\mathbb{Q}} - \bm{W}\bm{\beta} \right\|^2_2 + \lambda\left\|\bm{\beta}\right\|_1\right\}\\
&= \argmin_{\bm{\beta}\in\mathbb{R}^p} \left\{\left\|\widehat{\mathbb{Q}} - \bm{W}\bm{\beta} \right\|^2_2 + \lambda\left\|\bm{\beta}\right\|_1\right\}\\
&= \argmin_{\bm{\beta}\in\mathbb{R}^p} \left\{\left\| \bm{W}\bm{\beta}_0 + \widehat{\bm{\varepsilon}} - \bm{W}\bm{\beta} \right\|^2_2 + \lambda\left\|\bm{\beta}\right\|_1\right\}.
\end{align*}
Since $\widehat{\bm{\beta}}$ minimizes the objective function, we obtain the following inequality
\begin{align*}
&\left\| \bm{W}\bm{\beta}_0 + \widehat{\bm{\varepsilon}} - \bm{W}\widehat{\bm{\beta}}_n\right\|^2_2 + \lambda\left\|\widehat{\bm{\beta}}\right\|_1\leq \left\| \bm{W}\bm{\beta}_0 + \widehat{\bm{\varepsilon}} - \bm{W}\bm{\beta}_0 \right\|^2_2 + \lambda\left\|\bm{\beta}_0\right\|_1.
\end{align*}
Then, a simple calculation yields the statement.
\end{proof}

In the above statement, \citet{Buhlmann2011} refers to
\[2\widehat{\bm{\varepsilon}}^\top \bm{W}\left(\widehat{\bm{\beta}}_n- \bm{\beta}_0\right) / n,\]
as the \emph{empirical process} part. To develop an upper bound of $\| \bm{X}(\widehat{\bm{\beta}}_n- \bm{\beta}_0) \|^2_2$, we consider bounding the term. Let us decompose the empirical process term as 
\begin{align*}
    2\widehat{\bm{\varepsilon}}^\top \bm{W}\left(\widehat{\bm{\beta}}_n- \bm{\beta}_0\right) / n &= 2\overline{\bm{\varepsilon}}^\top \bm{W}\left(\widehat{\bm{\beta}}_n- \bm{\beta}_0\right) / n  - 2\overline{\bm{\varepsilon}}^\top \bm{W}\left(\widehat{\bm{\beta}}_n- \bm{\beta}_0\right) / n + 2\widehat{\bm{\varepsilon}}^\top \bm{W}\left(\widehat{\bm{\beta}}_n- \bm{\beta}_0\right) / n\\
    &\leq 2\overline{\bm{\varepsilon}}^\top \bm{W}\left(\widehat{\bm{\beta}}_n- \bm{\beta}_0\right) / n  + 2\left\|\overline{\bm{\varepsilon}}^\top \bm{W}\left(\widehat{\bm{\beta}}_n- \bm{\beta}_0\right) / n - \widehat{\bm{\varepsilon}}^\top \bm{W}\left(\widehat{\bm{\beta}}_n- \bm{\beta}_0\right) / n\right\|_{\infty}\\
    &\leq 2\overline{\bm{\varepsilon}}^\top \bm{W}\left(\widehat{\bm{\beta}}_n- \bm{\beta}_0\right) / n  + S,
\end{align*}
where recall that 
\begin{align*}
    S = 2\left\|\overline{\bm{\varepsilon}}^\top \bm{W}\left(\widehat{\bm{\beta}}_n- \bm{\beta}_0\right) / n - \widehat{\bm{\varepsilon}}^\top \bm{W}\left(\widehat{\bm{\beta}}_n- \bm{\beta}_0\right) / n\right\|_\infty.
\end{align*}
In the following part, we consider bounding $2\overline{\bm{\varepsilon}}^\top \bm{W}\left(\widehat{\bm{\beta}}_n- \bm{\beta}_0\right) / n$.

\subsection{Concentration Inequality Regarding the Oracle Error Term}
We consider bounding $(\max_{1\leq j \leq p} 2 |\overline{\bm{\varepsilon}}^\top \mathbb{W}_j|)\| \widehat{\bm{\beta}}_n- \bm{\beta}_0 \|_1$ for each $j\in[p]$.
We define the following event:
\begin{align*}
    \mathcal{F}_{\lambda_0} \coloneqq \left\{ \max_{1\leq j \leq p} 2 \left|\widehat{\bm{\varepsilon}}^\top \mathbb{W}_j\right| / n \leq \lambda_0\right\},
\end{align*}
for arbitrary $2\lambda_0 \leq \lambda$. Given appropriate $\lambda_0$ and sub-Gaussian errors $\overline{\bm{\varepsilon}} = (\widetilde{{\varepsilon}}_i)_{i=1}^n$, we prove that the event $\mathcal{F}_{\lambda_0}$ occurs with high probability. Let us define a Gram matrix as $\widehat{\Sigma} \coloneqq \bm{W}^\top \bm{W} / n$, and we denote its diagonal experiment by $\widehat{\sigma}^2_j$ for each $j\in[p]$. Then, we establish the following lemma.
\begin{lemma}[Corresponding to Lemma 6.2. in \citet{Buhlmann2011}]
\label{lem:concent}
Suppose that $\overline{\varepsilon}_i$ follows a sub-Gaussian distribution with a variance $\sigma^2_{\bar{\varepsilon}}$.
    Also suppose that $\widehat{\sigma}^2_j \leq M$ for all $j \in [p]$ and some $0 < M < \infty$. Then, for any $t > 0$ and for 
    \begin{align*}
        \lambda_0 \coloneqq 2M\sigma_{\bar{\varepsilon}}  \sqrt{\frac{t^2 + 2 \log(p)}{n}},
    \end{align*}
    we have
    \begin{align*}
        \mathbb{P}\left(\mathcal{F}_{\lambda_0}\right) \geq 1 - 2\exp\left( -t^2 / 2 \right).
    \end{align*}
\end{lemma}
\begin{proof}
    We first rewrite the term as
    \begin{align*}
        &\overline{\bm{\varepsilon}}^\top\mathbb{W}_j / \sqrt{n \sigma_{\bar{\varepsilon}}^2\widehat{\sigma}^2_j} = \sum^n_{i=1} \overline{\varepsilon}_iW_{i, j} / \sqrt{n \sigma_{\bar{\varepsilon}}^2\widehat{\sigma}^2_j}.
    \end{align*}
    Because $\overline{\varepsilon}_i$ follows a sub-Gaussian distribution with the fixed-design where recall that $D_i$ and $X_{i, j}$ are non-random,     
    the rewritten term also follows a sub-Gaussian distribution and achieve
    \begin{align*}
        &\mathbb{P}\left(\max_{1\leq j \leq p} \left|\overline{\bm{\varepsilon}}^\top \mathbb{W}_j\right| / \sqrt{n \sigma_{\bar{\varepsilon}}^2\widehat{\sigma}^2_j}\leq  \sqrt{t^2 + 2\log(p)}\right)\\
        &\leq \sum_{j\in[p]} \mathbb{P}\left(\left|\overline{\bm{\varepsilon}}^\top \mathbb{W}_j\right| / \sqrt{n \sigma_{\bar{\varepsilon}}^2\widehat{\sigma}^2_j}\leq  \sqrt{t^2 + 2\log(p)}\right)\\
        &\leq 2 p \exp\left( - \frac{t^2 + 2\log(p)}{2} \right) = 2 \exp\left(- \frac{t^2}{2}\right),
    \end{align*}
    where we used the definition of a sub-Gaussian random variable to derive the last inequality. 
\end{proof}

Here, from Lemma~\ref{lem:concent}, it holds that for any $t > 0$ and for 
    \begin{align*}
        \lambda \coloneqq M\sigma_{\bar{\varepsilon}} \sqrt{\frac{t^2 + 2 \log(p)}{n}}.
    \end{align*}
    with probability $1 - 2\exp(-t^2/2)$. Therefore,
\[\left\| \bm{W}\left(\widehat{\bm{\beta}}_n- \bm{\beta}_0\right) \right\|^2_2 / n \leq \left\| \bm{W}\left(\widehat{\bm{\beta}}_n- \bm{\beta}_0\right) \right\|^2_2 / n  + 2 \lambda  \left\| \widehat{\bm{\beta}}_n\right\|_1 \leq \lambda\left\|  \widehat{\bm{\beta}}_n- \bm{\beta}_0 \right\|_1 / 2 + \lambda \|\bm{\beta}_0\|_1,\]
with probability $1 - 2\exp(-t^2/2)$, which implies
\[\left\| \bm{W}\left(\widehat{\bm{\beta}}_n- \bm{\beta}_0\right) \right\|^2_2 / n \leq 3\lambda \|\bm{\beta}_0\|_1.\]

\subsection{Proof of Lemma~\ref{thm:oracle}}
\label{appdx:oracle_inequality}
In the proof of Lemma~\ref{thm:oracle}, the following lemma plays an important role. 
\begin{lemma}
\label{lem:oracle}
    Under the same conditions in Lemma~\ref{thm:oracle}, we have
    \begin{align*}
        \left\|\widehat{\bm{\beta}}_n- \bm{\beta}_0 \right\|^2_2 / n + \lambda\left\|\widehat{\bm{\beta}}_n- \bm{\beta}_0\right\|_1 \leq 4 \lambda^2 s_0 / \phi^2_0.
    \end{align*}
\end{lemma}
The statements in Lemma~\ref{thm:oracle} directly yield the statement.
\begin{proof}
Conditional on $\mathcal{F}_{\lambda_0}$, for $\lambda \geq 2 \lambda_0$, the basic inequality yields
\begin{align}
\label{eq:target1111}
    &2\left\| \widehat{\bm{\beta}}_n- \bm{\beta}_0 \right\|^2_2 / n + 2\lambda  \left\| \widehat{\bm{\beta}}_n\right\|_1 \leq \lambda\left\| \widehat{\bm{\beta}}_n- \bm{\beta}_0 \right\|_1 + 2\lambda \|\bm{\beta}_0\|_1 + S. 
\end{align}
For $\left\| \widehat{\bm{\beta}}_n\right\|_1$ in the LHS, from the triangle inequality and $\|\widehat{\bm{\beta}}_n\|_1 = \|\widehat{\bm{\beta}}_{n, \mathcal{S}_0}\|_1 + \|\widehat{\bm{\beta}}_{n, \mathcal{S}_0^c}\|_1$, the following holds:
\begin{align*}
    \left\| \widehat{\bm{\beta}}_n\right\|_1 &\geq \left\| \bm{\beta}_{0, \mathcal{S}_0}\right\|_1 - \left\| \widehat{\bm{\beta}}_{n, \mathcal{S}_0} - \bm{\beta}_{0, \mathcal{S}_0}\right\|_1 + \left\| \widehat{\bm{\beta}}_{n, \mathcal{S}^c_0}\right\|_1.
\end{align*}

For $\left\| \widehat{\bm{\beta}}_n- \bm{\beta}_0 \right\|_1$ in the RHS of \eqref{eq:target1111}, the following holds:
\begin{align*}
    \left\| \widehat{\bm{\beta}}_n- \bm{\beta}_0\right\|_1 &= \left\| \widehat{\bm{\beta}}_{n, \mathcal{S}_0} - \bm{\beta}_{0, \mathcal{S}_0} \right\|_1 + \left\| \widehat{\bm{\beta}}_{n, \mathcal{S}_0^c} \right\|_1.
\end{align*}
Here, from the definition of $\mathcal{S}^c_0$, $\|\bm{\beta}_{0, \mathcal{S}^c_0}\|_1 = 0$ holds. Therefore, we have
 \begin{align*}
        &2\left\| \widehat{\bm{\beta}}_n- \bm{\beta}_0 \right\|^2_2 / n +2\lambda\left\| \bm{\beta}_{0, \mathcal{S}_0}\right\|_1 - 2\lambda\left\| \bm{\beta}_{0, \mathcal{S}_0} - \widehat{\bm{\beta}}_{n, \mathcal{S}_0} \right\|_1  + 2\lambda\left\| \widehat{\bm{\beta}}_{n, \mathcal{S}^c_0}\right\|_1\\
        & \leq \lambda\left\| \widehat{\bm{\beta}}_{n, \mathcal{S}_0} - \bm{\beta}_{0, \mathcal{S}_0} \right\|_1 +  \lambda\left\|\widehat{\bm{\beta}}_{n, \mathcal{S}^c_0} \right\|_1 + 2\lambda\left\| \bm{\beta}_{0, \mathcal{S}_0}\right\|_1 + 2\lambda\left\| \bm{\beta}_{0, \mathcal{S}^c_0}\right\|_1 + S\\
        & = \lambda\left\| \widehat{\bm{\beta}}_{n, \mathcal{S}_0} - \bm{\beta}_{0, \mathcal{S}_0} \right\|_1 +  \lambda\left\|\widehat{\bm{\beta}}_{n, \mathcal{S}^c_0} \right\|_1 + 2\lambda\left\| \bm{\beta}_{0, \mathcal{S}_0}\right\|_1 + S.
\end{align*}
Then, the following holds:
  \begin{align*}
        &2\left\| \bm{W}\left(\widehat{\bm{\beta}}_n- \bm{\beta}_0\right) \right\|^2_2 / n + \lambda  \left\| \widehat{\bm{\beta}}_{n, \mathcal{S}^c_0}\right\|_1 \leq 3\lambda\left\| \widehat{\bm{\beta}}_{n, \mathcal{S}_0} - \bm{\beta}_{0, \mathcal{S}_0} \right\|_1 + S.
 \end{align*}

Next, we consider bounding $2\left\|\bm{W}\left(\widehat{\bm{\beta}}_n- \bm{\beta}_0\right) \right\|^2_2 / n + \lambda\left\|\widehat{\bm{\beta}}_n- \bm{\beta}_0\right\|_1$ as
    \begin{align*}
        &2\left\|\bm{W}\left(\widehat{\bm{\beta}}_n- \bm{\beta}_0\right) \right\|^2_2 / n + \lambda\left\|\widehat{\bm{\beta}}_n- \bm{\beta}_0\right\|_1\\
        &\leq 2\left\|\bm{W}\left(\widehat{\bm{\beta}}_n- \bm{\beta}_0\right) \right\|^2_2 / n + \lambda\left\|\widehat{\bm{\beta}}_{n, \mathcal{S}^c_0}\right\|_1 + \lambda\left\|\widehat{\bm{\beta}}_{n, \mathcal{S}_0} - \bm{\beta}_{0, \mathcal{S}_0}\right\|_1\\
        &\leq 4\lambda\left\|\widehat{\bm{\beta}}_{n, \mathcal{S}_0} - \bm{\beta}_{0, \mathcal{S}_0}\right\|_1.
    \end{align*}
Lastly, from the compatibility condition $\left\|\widehat{\bm{\beta}}_{n, \mathcal{S}_0} - \bm{\beta}_{0, \mathcal{S}_0}\right\|^2_1 \leq s_0 \left(\widehat{\bm{\beta}}_{n, \mathcal{S}_0} - \bm{\beta}_{0, \mathcal{S}_0}\right)^\top \widehat{\Sigma} \left(\widehat{\bm{\beta}}_{n, \mathcal{S}_0} - \bm{\beta}_{0, \mathcal{S}_0}\right) / \phi^2_0$, we have
\begin{align*}
    &4\lambda\left\|\widehat{\bm{\beta}}_{n, \mathcal{S}_0} - \bm{\beta}_{0, \mathcal{S}_0}\right\|_1 + S\\
    &\leq  4\lambda \sqrt{s_0}\left\|\bm{W}\left(\widehat{\bm{\beta}} - \bm{\beta}_0\right) \right\|_2  / \left( \sqrt{n} \phi_0 \right) + S\\
    &\leq  \left\|\bm{W}\left(\widehat{\bm{\beta}}_n- \bm{\beta}_0\right) \right\|^2_2  / n + 4\lambda^2 s_0 / \phi^2_0 + S,
\end{align*}
where in the last inequality, we used $4 uv \leq u^2 + 4v^2$. Therefore, we obtain 
\begin{align*}
    2\left\|\bm{W}\left(\widehat{\bm{\beta}}_n- \bm{\beta}_0\right) \right\|^2_2 / n + \lambda\left\|\widehat{\bm{\beta}}_n- \bm{\beta}_0\right\|_1 \leq \left\|\bm{W}\left(\widehat{\bm{\beta}}_n- \bm{\beta}_0\right) \right\|^2_2  / n + 4\lambda^2 s_0 / \phi^2_0 + S.
\end{align*}
Thus, from this inequality, the statement holds. 
\end{proof}

\section{Proof of Lemma~\ref{lem:conv_Q}}
\label{appdx:lem:conv_Q}
Recall that we defined 
\begin{align*}
    \overline{u}_i &\coloneqq \frac{\mathbbm{1}\big[D_i \coloneqq 1\big]\big(Y_i - \mu_0(1)(X_i)\big)}{\pi_0(1\mid X_i)} - \frac{\mathbbm{1}\big[D_i = 0\big]\big(Y_i - \mu_0(0)(X_i)\big)}{\pi_0(0\mid X_i)} - \Big\{Y_i(1) - \mu(1)(X_i)\Big\} + \Big\{Y_i(0) - \mu(0)(X_i)\Big\},\\
     \widehat{u}^{\mathrm{DML}}_i & \coloneqq \frac{\mathbbm{1}\big[D_i = 1\big]\big(Y_i - \widehat{\mu}_{\ell(i), n}(1)(X_i)\big)}{\widehat{\pi}_{\ell(i), n}(1\mid X_i)} - \frac{\mathbbm{1}\big[D_i = 0\big]\big(Y_i - \widehat{\mu}_{\ell(i), n}(0)(X_i)\big)}{\widehat{\pi}_{\ell(i), n}(0\mid X_i)}\\
     &\ \ \ \ \ \ \ \ \  - \Big\{Y_i(1) - \widehat{\mu}_{\ell(i), n}(1)(X_i)\Big\}+ \Big\{Y_i(0) - \widehat{\mu}_{\ell(i), n}(0)(X_i)\Big\}. 
\end{align*}
In the following lemma, we show that $\sqrt{n}\big(\overline{u}_i - \widehat{u}^{\mathrm{DML}}_i\big) = o_P(1)$ holds. This lemma directly yields Lemma~\ref{lem:conv_Q}. 
\begin{lemma}
\label{lem:nuisance}
The following holds:
\begin{align}
\label{eq:target2}
    \sqrt{n}\left(\overline{u}_i - \widehat{u}^{\mathrm{DML}}_i\right) = o_P(1).
\end{align}
\end{lemma}
\begin{proof}
To show \eqref{eq:target2}, we use the following decomposition:
\begin{align*}
    &\sqrt{n}\left(\overline{u}_i - \widehat{u}^{\mathrm{DML}}_i\right)\\
    &= \sqrt{n}\Bigg\{\frac{\mathbbm{1}\big[D_i \coloneqq 1\big]\big(Y_i - \mu_0(1)(X_i)\big)}{\pi_0(1\mid X_i)} - \frac{\mathbbm{1}\big[D_i = 0\big]\big(Y_i - \mu_0(0)(X_i)\big)}{\pi_0(0\mid X_i)} - \Big\{Y_i(1) - \mu_0(1)(X_i)\Big\} + \Big\{Y_i(0) - \mu_0(0)(X_i)\Big\}\\
    &\ \ \ \ \ \ \ \ \ - \frac{\mathbbm{1}\big[D_i = 1\big]\big(Y_i - \widehat{\mu}_{\ell(i), n}(1)(X_i)\big)}{\widehat{\pi}_{\ell(i), n}(1\mid X_i)} + \frac{\mathbbm{1}\big[D_i = 0\big]\big(Y_i - \widehat{\mu}_{\ell(i), n}(0)(X_i)\big)}{\widehat{\pi}_{\ell(i), n}(0\mid X_i)}\\
     &\ \ \ \ \ \ \ \ \ + \Big\{Y_i(1) - \widehat{\mu}_{\ell(i), n}(1)(X_i)\Big\} - \Big\{Y_i(0) - \widehat{\mu}_{\ell(i), n}(0)(X_i)\Big\}\Bigg\}\\
     &= \sqrt{n}\Bigg\{\frac{\mathbbm{1}\big[D_i \coloneqq 1\big]\big(Y_i - \mu_0(1)(X_i)\big)}{\pi_0(1\mid X_i)} - \frac{\mathbbm{1}\big[D_i = 0\big]\big(Y_i - \mu_0(0)(X_i)\big)}{\pi_0(0\mid X_i)}\\
    &\ \ \ \ \ \ \ \ \ - \frac{\mathbbm{1}\big[D_i = 1\big]\big(Y_i - \widehat{\mu}_{\ell(i), n}(1)(X_i)\big)}{\widehat{\pi}_{\ell(i), n}(1\mid X_i)} + \frac{\mathbbm{1}\big[D_i = 0\big]\big(Y_i - \widehat{\mu}_{\ell(i), n}(0)(X_i)\big)}{\widehat{\pi}_{\ell(i), n}(0\mid X_i)}\\
     &\ \ \ \ \ \ \ \ \ + \Big\{\mu_0(1)(X_i) - \widehat{\mu}_{\ell(i), n}(1)(X_i)\Big\} - \Big\{\mu_0(0)(X_i) - \widehat{\mu}_{\ell(i), n}(0)(X_i)\Big\}\Bigg\}\\
     &= \sqrt{n}\Bigg\{\frac{\mathbbm{1}\big[D_i \coloneqq 1\big]\big(Y_i - \mu_0(1)(X_i)\big)}{\pi_0(1\mid X_i)} - \frac{\mathbbm{1}\big[D_i = 0\big]\big(Y_i - \mu_0(0)(X_i)\big)}{\pi_0(0\mid X_i)}\\
    &\ \ \ \ \ \ \ \ \ - \frac{\mathbbm{1}\big[D_i = 1\big]\big(Y_i - \widehat{\mu}_{\ell(i), n}(1)(X_i)\big)}{\widehat{\pi}_{\ell(i), n}(1\mid X_i)} + \frac{\mathbbm{1}\big[D_i = 0\big]\big(Y_i - \widehat{\mu}_{\ell(i), n}(0)(X_i)\big)}{\widehat{\pi}_{\ell(i), n}(0\mid X_i)}\\
     &\ \ \ \ \ \ \ \ \ + \Big\{\mu_0(1)(X_i) - \widehat{\mu}_{\ell(i), n}(1)(X_i)\Big\} - \Big\{\mu_0(0)(X_i) - \widehat{\mu}_{\ell(i), n}(0)(X_i)\Big\}\Bigg\}\\
     &\ \ \ - \sqrt{n}\mathbb{E}\Bigg[\frac{\mathbbm{1}\big[D_i \coloneqq 1\big]\big(Y_i - \mu_0(1)(X_i)\big)}{\pi_0(1\mid X_i)} - \frac{\mathbbm{1}\big[D_i = 0\big]\big(Y_i - \mu_0(0)(X_i)\big)}{\pi_0(0\mid X_i)}\\
    &\ \ \ \ \ \ \ \ \ - \frac{\mathbbm{1}\big[D_i = 1\big]\big(Y_i - \widehat{\mu}_{\ell(i), n}(1)(X_i)\big)}{\widehat{\pi}_{\ell(i), n}(1\mid X_i)} + \frac{\mathbbm{1}\big[D_i = 0\big]\big(Y_i - \widehat{\mu}_{\ell(i), n}(0)(X_i)\big)}{\widehat{\pi}_{\ell(i), n}(0\mid X_i)}\\
     &\ \ \ \ \ \ \ \ \ + \Big\{\mu_0(1)(X_i) - \widehat{\mu}_{\ell(i), n}(1)(X_i)\Big\} - \Big\{\mu_0(0)(X_i) - \widehat{\mu}_{\ell(i), n}(0)(X_i)\Big\}\Bigg]\\
     &\ \ \ + \sqrt{n}\mathbb{E}\Bigg[\frac{\mathbbm{1}\big[D_i \coloneqq 1\big]\big(Y_i - \mu_0(1)(X_i)\big)}{\pi_0(1\mid X_i)} - \frac{\mathbbm{1}\big[D_i = 0\big]\big(Y_i - \mu_0(0)(X_i)\big)}{\pi_0(0\mid X_i)}\\
    &\ \ \ \ \ \ \ \ \ - \frac{\mathbbm{1}\big[D_i = 1\big]\big(Y_i - \widehat{\mu}_{\ell(i), n}(1)(X_i)\big)}{\widehat{\pi}_{\ell(i), n}(1\mid X_i)} + \frac{\mathbbm{1}\big[D_i = 0\big]\big(Y_i - \widehat{\mu}_{\ell(i), n}(0)(X_i)\big)}{\widehat{\pi}_{\ell(i), n}(0\mid X_i)}\\
     &\ \ \ \ \ \ \ \ \ + \Big\{\mu_0(1)(X_i) - \widehat{\mu}_{\ell(i), n}(1)(X_i)\Big\} - \Big\{\mu_0(0)(X_i) - \widehat{\mu}_{\ell(i), n}(0)(X_i)\Big\}\Bigg].
\end{align*}

For a measurable function $\psi_1:\mathcal{Y}\times\{1, 0\}\mathcal{X}\times \mathcal{M}\times\Pi\to\mathbb{R}$ and $\psi_2:\mathcal{X}\times \mathcal{M}\to\mathbb{R}$, let us define the empirical expectations over the labeled and unlabeled data as
\begin{align*}
&\widehat{\mathbb{E}}_{n_{\ell}}[\psi_1(Y, D, X)] \coloneqq \frac{1}{n_\ell}\sum_{i\in\mathcal{D}^\ell}\psi_1(Y_i, D_i, X_i),\qquad \widehat{\mathbb{E}}_{n_{\ell}}[\psi_2(X)] \coloneqq \frac{1}{n_\ell}\sum_{i\in\mathcal{D}^\ell}\psi_2(X_i).
\end{align*} 

We define 
\begin{align*}
\phi_1(Y, D, X; \mu, \pi) &\coloneqq \frac{\mathbbm{1}\big[D = 1\big]\big(Y - \mu(1)(X)\big)}{\pi(1\mid X)} - \frac{\mathbbm{1}\big[D = 0\big]\big(Y - \mu(0)(X)\big)}{\pi(0\mid X)},\\
\phi_2(X; f) &\coloneqq \mu(1)(X) - \mu(0)(X).
\end{align*}

Let $\mathbb{G}_{n_\ell}$ be an empirical process for functions $\psi_1:\mathcal{Y}\times\{1, 0\}\mathcal{X}\times \mathcal{M}\times\Pi\to\mathbb{R}$ and $\psi_2:\mathcal{X}\times \mathcal{M}\to\mathbb{R}$ defined as 
\begin{align*}
&\mathbb{G}_{n_\ell}(\psi_1(Y, D, X; \mu, \pi))=\sqrt{n_\ell}\Big(\widehat{\mathbb{E}}_{n_{\ell}}[\psi_1(Y, D, X; \mu, \pi)] - \mathbb{E}_{P_0}[\psi_1(X, Y; \mu, \pi)]\Big)\\
&\mathbb{G}_{n_\ell}(\psi_2(X_i; \mu))=\sqrt{n_\ell}\Big(\widehat{\mathbb{E}}_{n_{\ell}}[\psi_2(X; \mu)] - \mathbb{E}_{P_0}[\psi_2(X; \mu)]\Big).
\end{align*}

Then, we have
\begin{align*}
    &\sqrt{n}\sum_{i\in\mathcal{D}^\ell}\left(\overline{u}_i - \widehat{u}^{\mathrm{DML}}_i\right)\\
    &= \frac{\sqrt{n}}{\sqrt{n_\ell}}\mathbb{G}_{n_\ell}\left(\phi_1\Big(Y, D, X; \mu_0, \pi_0\Big) - \phi_1\Big(Y, D, X; \widehat{\mu}_{\ell, n}, \widehat{\pi}_{\ell, n}\Big) \right) +\frac{\sqrt{n}}{\sqrt{n_\ell}}\mathbb{G}_{n_\ell}\left( \phi_2\Big(X; \mu_0\Big) - \phi_2\Big(X; \widehat{\mu}_{\ell(i), i}\Big)\right)\\
    &\ \ \ + \sqrt{n}\Bigg\{\mathbb{E}\Big[\phi_1\Big(Y, D, X; \widehat{\mu}_{\ell, n}, \widehat{\pi}_{\ell, n}\Big)\mid \widehat{\mu}_{\ell, n}, \widehat{\pi}_{\ell, n}\Big] - \mathbb{E}\Big[\phi_1\Big(Y, D, X; \mu_0, \pi_0\Big)\mid \widehat{\mu}_{\ell, n}, \widehat{\pi}_{\ell, n}\Big]\nonumber\\
&\ \ \ \ \ \ \ \ \ \ \ \ \ \ \ \ \ \ \ \ \ \ \ \ \ \ \ \ \ \ \ \ \ \ \ \ \ \ \ \ \ \ \ \ \ \ \ \ \ \ \ \ \ \ \ \ \ \ \ \ \ \ \ \ + \mathbb{E}\Big[\phi_2\Big(X; \mu_0\Big)\mid \widehat{\mu}_{\ell, n}\Big] - \mathbb{E}\Big[\phi_2\Big(X; \widehat{\mu}_{\ell(i), i}\Big)\mid \widehat{\mu}_{\ell, n}\Big]\Bigg\}
\end{align*}

Then, in the following proof, we separately show the following two inequalities:
\begin{align}
\label{eq:target4}
&\frac{\sqrt{n}}{\sqrt{n_\ell}}\mathbb{G}_{n_\ell}\left(\phi_1\Big(Y, D, X; \mu_0, \pi_0\Big) - \phi_1\Big(Y, D, X; \widehat{\mu}_{\ell, n}, \widehat{\pi}_{\ell, n}\Big) \right)\nonumber\\
&\ \ \ \ \ \ \ \ \ \ \ \ \ \ \ \ \ \ \ \ \ \ \ \ \ \ \ \ \ \ \ \ \ \ \ \ \ \ \ \ \ +\frac{\sqrt{n}}{\sqrt{n_\ell}}\mathbb{G}_{n_\ell}\left( \phi_2\Big(X; \mu_0\Big) - \phi_2\Big(X; \widehat{\mu}_{\ell(i), i}\Big)\right) = o_P(1),
\end{align}
and 
\begin{align}
\label{eq:target5}
&\sqrt{n}\Bigg\{\mathbb{E}\Big[\phi_1\Big(Y, D, X; \widehat{\mu}_{\ell, n}, \widehat{\pi}_{\ell, n}\Big)\mid \widehat{\mu}_{\ell, n}, \widehat{\pi}_{\ell, n}\Big] - \mathbb{E}\Big[\phi_1\Big(Y, D, X; \mu_0, \pi_0\Big)\mid \widehat{\mu}_{\ell, n}, \widehat{\pi}_{\ell, n}\Big]\nonumber\\
&\ \ \ \ \ \ \ \ \ \ \ \ \ \ \ \ \ \ \ \ \ \ \ \ \ \ \ \ \ \ \ \ \ \ \ \ \ \ \ \ \ \ \ \ \ \ \ \ \ \ \ \ \ \ \ \ \ \ \ \ \ \ \ \ + \mathbb{E}\Big[\phi_2\Big(X; \widehat{\mu}_{\ell(i), i}\Big)\mid \widehat{\mu}_{\ell, n}\Big]- \mathbb{E}\Big[\phi_2\Big(X; \mu_0\Big)\mid \widehat{\mu}_{\ell, n}\Big]\Bigg\}\\
&= o_P(1)\nonumber.
\end{align}

\paragraph{Step~1: Proof of \eqref{eq:target4}.}

If we can show that for any $\epsilon>0$, 
\begin{align}\
\label{eq:part}
 &\lim_{n\to \infty}\mathbb{P}\Bigg[\Bigg|\frac{\sqrt{n}}{\sqrt{n_\ell}}\mathbb{G}_{n_\ell}\big(\phi_1(Y, D, X; \widehat{\mu}_{\ell, n}, \widehat{\pi}_{\ell, n}) - \phi_1\Big(Y, D, X; \mu_0, \pi_0\Big)\big)\nonumber\\
&\ \ \ \ \ \ \ \ \ \ \ \ \ \ \ \ \ \ \ \ \ \ \ \ +\frac{\sqrt{n}}{\sqrt{n_\ell}}\mathbb{G}_{n_\ell}\big( \phi_2\Big(X; \mu_0\Big) - \phi_2\Big(X; \widehat{\mu}_{\ell(i), i}\Big)\big) \Bigg| > \varepsilon \mid D_{(2)} \Bigg]=0,
\end{align}
then by the bounded convergence theorem, we would have 
\begin{align*}
&\lim_{n \to \infty}P\Bigg[\Bigg|\frac{\sqrt{n}}{\sqrt{n_\ell}}\mathbb{G}_{n_\ell}\big(\phi_1\Big(Y, D, X; \mu_0, \pi_0\Big) - \phi_1\Big(Y, D, X; \widehat{\mu}_{\ell, n}, \widehat{\pi}_{\ell, n}\Big) \big)\nonumber\\
&\ \ \ \ \ \ \ \ \ \ \ \ \ \ \ \ \ \ \ \ \ \ \ \ +\frac{\sqrt{n}}{\sqrt{n_\ell}}\mathbb{G}_{n_\ell}\big( \phi_2\Big(X; \mu_0\Big) - \phi_2\Big(X; \widehat{\mu}_{\ell(i), i}\Big)\big) \Bigg| > \varepsilon \mid D_{(2)} \Bigg]=0,
\end{align*}
yielding the statement. 

To show \eqref{eq:part}, we show that the conditional mean is $0$ and the conditional variance is $o_P(1)$. Then, \eqref{eq:part} is proved by the Chebyshev inequality following the proof of \citep[Theorem 4]{KallusUehara2019}. 
The conditional mean is 
\begin{align*}
&\mathbb{E}\Bigg[\frac{\sqrt{n}}{\sqrt{n_\ell}}\mathbb{G}_{n_\ell}\big(\phi_1(Y, D, X; \widehat{\mu}_{\ell, n}, \widehat{\pi}_{\ell, n}) - \phi_1\Big(Y, D, X; \mu_0, \pi_0\Big)\big)\\
&\ \ \ \ \ +\frac{\sqrt{n}}{\sqrt{n_\ell}}\mathbb{G}_{n_\ell}\big(\phi_2(X; \widehat{\mu}_{\ell, n}) - \phi_2\Big(X; \mu_0\Big)\big) \mid \overline{D}^{\ell}\Bigg] \\
&= \mathbb{E}\Bigg[\frac{\sqrt{n}}{\sqrt{n_\ell}}\mathbb{G}_{n_\ell}\big(\phi_1(Y, D, X; \widehat{\mu}_{\ell, n}, \widehat{\pi}_{\ell, n}) - \phi_1\Big(Y, D, X; \mu_0, \pi_0\Big)\big),\\
&\ \ \ \ \ +\frac{\sqrt{n}}{\sqrt{n_\ell}}\mathbb{G}_{n_\ell}\big(\phi_2(X; \widehat{\mu}_{\ell, n}) - \phi_2\Big(X; \mu_0\Big)\big) \mid \widehat{\mu}_{\ell, n}, \widehat{\pi}_{\ell, n}\Bigg] \\
&=0. 
\end{align*}

The conditional variance is bounded as 
\begin{align*}
&\mathrm{Var}\Bigg[\frac{\sqrt{n}}{\sqrt{n_\ell}}\mathbb{G}_{n_\ell}\big(\phi_1(Y, D, X; \widehat{\mu}_{\ell, n}, \widehat{\pi}_{\ell, n}) - \phi_1\Big(Y, D, X; \mu_0, \pi_0\Big)\big),\\
&\ \ \ \ \ +\frac{\sqrt{n}}{\sqrt{n_\ell}}\mathbb{G}_{n_\ell}\big(\phi_2(X; \widehat{\mu}_{\ell, n}) - \phi_2\Big(X; \mu_0\Big)\big) \mid \overline{D}^{\ell}\Bigg]\\
&=\frac{n}{n_\ell}\mathrm{Var}\Bigg[\phi_1(Y, D, X; \widehat{\mu}_{\ell, n}, \widehat{\pi}_{\ell, n}) - \phi_1\Big(Y, D, X; \mu_0, \pi_0\Big)\mid \overline{D}^{\ell}\Bigg]\\
&\ \ \ \ \ +\frac{n}{n_\ell}\mathrm{Var}\Bigg[\phi_2(X; \widehat{\mu}_{\ell, n}) - \phi_2\Big(X; \mu_0\Big) \mid \overline{D}^{\ell}\Bigg]\\
&=\frac{n}{n_\ell}\mathbb{E}\Bigg[\left\{\phi_1(Y, D, X; \widehat{\mu}_{\ell, n}, \widehat{\pi}_{\ell, n}) - \phi_1\Big(Y, D, X; \mu_0, \pi_0\Big)\right\}^2 \mid \overline{D}^{\ell}\Bigg]\\
&\ \ \ \ \ +\frac{n}{n_\ell}\mathbb{E}\Bigg[\left\{\phi_2(X; \widehat{\mu}_{\ell, n}) - \phi_2\Big(X; \mu_0\Big)\right\}^2  \mid \overline{D}^{\ell}\Bigg]\\
&=o_P(1)
\end{align*}

Here, we used
\begin{align}
\label{eq:first_e}
    \frac{n}{n_\ell}\mathbb{E}\Bigg[\left\{\phi_1\Big(Y, D, X; \mu_0, \pi_0\Big) - \phi_1\Big(Y, D, X; \widehat{\mu}_{\ell, n}, \widehat{\pi}_{\ell, n}\Big) \right\}^2 \mid \overline{D}^{\ell}\Bigg]=o_P(1),
\end{align}
and 
\begin{align} 
\label{eq:second_e}
    \frac{n}{n_\ell}\mathbb{E}\Bigg[\left\{\phi_2(X_i; \widehat{\mu}_{\ell, n}) - \phi_2\Big(X; \mu_0\Big)\right\}^2  \mid \overline{D}^{\ell}\Bigg]=o_P(1). 
\end{align}
The first equation \eqref{eq:first_e} is proved by 
\begin{align*}
    &\mathbb{E}\Bigg[\Bigg\{\left(\frac{\mathbbm{1}\big[D = 1\big]\big(Y - \widehat{\mu}_{\ell, n}(1)(X)\big)}{\widehat{\pi}_{\ell, n}(1\mid X)} - \frac{\mathbbm{1}\big[D = 0\big]\big(Y - \widehat{\mu}_{\ell, n}(0)(X)\big)}{\widehat{\pi}_{\ell, n}(0\mid X)}\right)\\
    &\ \ \ \ \ - \left(\frac{\mathbbm{1}\big[D = 1\big]\big(Y - \mu_0(1)(X)\big)}{\pi_0(1\mid X)} - \frac{\mathbbm{1}\big[D = 0\big]\big(Y - \mu_0(0)(X)\big)}{\pi_0(0\mid X)}\right)\Bigg\}^2 \mid \overline{D}^{\ell}\Bigg]\\
    &=\mathbb{E}\Bigg[\Bigg\{\left(\frac{\mathbbm{1}\big[D = 1\big]\big(Y - \widehat{\mu}_{\ell, n}(1)(X)\big)}{\widehat{\pi}_{\ell, n}(1\mid X)} - \frac{\mathbbm{1}\big[D = 0\big]\big(Y - \widehat{\mu}_{\ell, n}(0)(X)\big)}{\widehat{\pi}_{\ell, n}(0\mid X)}\right)\\
    &\ \ \ \ \ - \left(\frac{\mathbbm{1}\big[D = 1\big]\big(Y - \mu_0(1)(X)\big)}{\widehat{\pi}_{\ell, n}(1\mid X)} - \frac{\mathbbm{1}\big[D = 0\big]\big(Y - \mu_0(0)(X)\big)}{\widehat{\pi}_{\ell, n}(0\mid X)}\right)\\
    &\ \ \ \ \ + \left(\frac{\mathbbm{1}\big[D = 1\big]\big(Y - \mu_0(1)(X)\big)}{\widehat{\pi}_{\ell, n}(1\mid X)} - \frac{\mathbbm{1}\big[D = 0\big]\big(Y - \mu_0(0)(X)\big)}{\widehat{\pi}_{\ell, n}(0\mid X)}\right)\\
    &\ \ \ \ \ - \left(\frac{\mathbbm{1}\big[D = 1\big]\big(Y - \mu_0(1)(X)\big)}{\pi_0(1\mid X)} - \frac{\mathbbm{1}\big[D = 0\big]\big(Y - \mu_0(0)(X)\big)}{\pi_0(0\mid X)}\right)\Bigg\}^2 \mid \overline{D}^{\ell}\Bigg]\\
    &\leq C_1 \big\| \mu_0(a)(X) - \widehat{\mu}_{\ell, n}(a)(X)\big\|^2_2 + C_2 \big\|\widehat{\pi}_{\ell, n}(a\mid X) - \pi_0(a\mid X) \big\|^2_2,
\end{align*}
where $C_1, C_2$ are some constants independent of $n$. 
Here, we have used $\widetilde{C}_1 <\widehat \pi_0$ and $|Y|, |\widehat{\mu}_{\ell, n}|<\widetilde{C}_2$ for some constants $\widetilde{C}_1, \widetilde{C}_2 > 0$ according to the Assumptions \ref{asm:bounded_output} and \ref{asmp:coherent}. Then, from the convergence rate condition (Assumption~\ref{asm:conv_rate}), we have \eqref{eq:first_e}.

The second equation \eqref{eq:second_e} directly holds from Assumption~\ref{asm:conv_rate}.

\paragraph{Step~2: Proof of \eqref{eq:target5}.}

We have 
\begin{align*}
&\Bigg|\mathbb{E}\Big[\phi_1\Big(Y, D, X; \mu_0, \pi_0\Big) - \mathbb{E}\Big[\phi_1\Big(Y, D, X; \widehat{\mu}_{\ell, n}, \widehat{\pi}_{\ell, n}\Big)\mid \widehat{\mu}_{\ell, n}, \widehat{\pi}_{\ell, n}\Big]\mid \widehat{\mu}_{\ell, n}, \widehat{\pi}_{\ell, n}\Big]\\
&\ \ \ \ \ + \mathbb{E}\Big[\phi_2\Big(X; \mu_0\Big)\mid \widehat{\mu}_{\ell, n}\Big] - \mathbb{E}\Big[\phi_2(X_i; \widehat{\mu}_{\ell, n})\mid \widehat{\mu}_{\ell, n}\Big]\Bigg|\\
&= \Bigg|\mathbb{E}\Bigg[\left(\frac{\mathbbm{1}\big[D = 1\big]\big(Y - \mu_0(1)(X)\big)}{\pi_0(1\mid X)} - \frac{\mathbbm{1}\big[D = 0\big]\big(Y - \mu_0(0)(X)\big)}{\pi_0(0\mid X)}\right)\\
    &\ \ \ \ \ - \left(\frac{\mathbbm{1}\big[D = 1\big]\big(Y - \widehat{\mu}_{\ell, n}(1)(X)\big)}{\widehat{\pi}_{\ell, n}(1\mid X)} - \frac{\mathbbm{1}\big[D = 0\big]\big(Y - \widehat{\mu}_{\ell, n}(0)(X)\big)}{\widehat{\pi}_{\ell, n}(0\mid X)}\right)\mid \widehat{\mu}_{\ell, n}, \widehat{\pi}_{\ell, n}\Bigg]\\
&\ \ \ \ \ + \mathbb{E}\Big[\Big\{\mu_0(1)(X_i) - \widehat{\mu}_{\ell(i), n}(1)(X_i)\Big\} - \Big\{\mu_0(0)(X_i) - \widehat{\mu}_{\ell(i), n}(0)(X_i)\Big\}\mid \widehat{\mu}_{\ell, n}\Big]\Bigg|\\
&= \Bigg| - \mathbb{E}\Bigg[\left(\frac{\pi_0(1\mid X)\big(\mu_0(1)(X) - \widehat{\mu}_{\ell, n}(1)(X)\big)}{\widehat{\pi}_{\ell, n}(1\mid X)} - \frac{\pi_0(0\mid X)\big(\mu_0(0)(X) - \widehat{\mu}_{\ell, n}(0)(X)\big)}{\widehat{\pi}_{\ell, n}(0\mid X)}\right)\mid \widehat{\mu}_{\ell, n}, \widehat{\pi}_{\ell, n}\Bigg]\\
&\ \ \ \ \ + \mathbb{E}\Big[\Big\{\mu_0(1)(X_i) - \widehat{\mu}_{\ell(i), n}(1)(X_i)\Big\} - \Big\{\mu_0(0)(X_i) - \widehat{\mu}_{\ell(i), n}(0)(X_i)\Big\}\mid \widehat{\mu}_{\ell, n}\Big]\Bigg|\\
&= C\sum_{d\in\{1, 0\}}\Bigg|\mathbb{E}\Big[\big(\mu_0(d)(X)-\widehat{\mu}_{\ell, n}(d)(X)\big)\big(\widehat{\pi}_{\ell, n}(d)(X) - \pi_0(d\mid X)\big)\Big]\Bigg|\\
&=\mathrm{o}_p(n^{-1/2}),
\end{align*}
where $C > 0$ is some constant. 
All elements of the vector is bounded by a universal constant.
Here, we have used \Holder's inequality $ \|fg \|_1 \leq  \|f \|_2  \|g \|_2$.

\end{proof}

\section{Proof of Lemma~\ref{lem:asymp_normal}}
\label{appdx:nongaussian}
Let $\tau^2_j$ be the population value of estimates $\widehat{\tau}^2_j$ in the nodewise Lasso, defined as
\begin{align*}
    \tau^2_j = \mathbb{E}_{P_0}\left[\left(W_{1, j}- \sum_{k\neq j} \gamma_{0,j,k}W_{1, k}\right)^2\right],
\end{align*}
where $\{\gamma_{0,j,k}\}_{k\neq j}$ is the population regression coefficient of $W_{1, j}$ versus $\{W_{1, k}\}_{k\neq j}$ that satisfy the following conditions: uniformly in $j\in[p]$, $\tau^2_j = 1 / \Theta_{j, j} \geq \Lambda^2_{\min} > 0$ and $\tau^2_j \leq \mathbb{E}_{P_0}\left[W_{1, j}\right] = \Sigma_{j, j} = O(1)$. 

In the proof, we use the following proposition from \citet{vandeGeer2014}. 
\begin{proposition}[From Theorem~2.4 in \citet{vandeGeer2014}]
\label{lem:random}
        Denote by $\Theta \coloneqq \Theta_{\mathrm{nodewise}}$.
Suppose that Assumptions~\ref{asm:b1_vandegeer}--\ref{asm:b4_vandegeer} hold. Then, for suitable tuning parameters $\lambda_j \asymp K_0 \sqrt{\log(p)/n}$ uniformly in $j$, 
    \begin{align*}
        \left\| \Theta_{\mathrm{nodewise}, j} - \Theta_j \right\|_1 &= o_P\left( K_0s_j \sqrt{\frac{\log(p)}{n}} \right),\\
        \left\| \Theta_{\mathrm{nodewise}, j} - \Theta_j \right\|_2 &= o_P\left( K_0s_j \sqrt{\frac{\log(p)}{n}} \right),\\
        \left| \widehat{\tau}^2_j - \tau^2_j \right| &= o_P\left( K_0\sqrt{\frac{s_j \log(p)}{n}} \right),\qquad j \in [p].
    \end{align*}
Furthermore, we have
\begin{align*}
    &\left| \Theta_{\mathrm{nodewise}, j} \Sigma \left(\Theta_{\mathrm{nodewise}, j}\right)^\top - \Theta^d_{j, j} \right|\\
    &\leq \left\| \Sigma \right\|_{\infty} \left\| \Theta_{\mathrm{nodewise}, j} - \Theta_j \right\|^2_1 \lor \Lambda^2_{\max}\left\| \Theta_{\mathrm{nodewise}, j} - \Theta_j \right\|^2_2 + 2\left| \widehat{\tau}^2_j - \tau^2_j \right|,\qquad j \in[p],
\end{align*}
where $\Lambda^2_{\max}$ is the maximal eigenvalue of $\Sigma$. In the sub-Gaussian or strongly bounded case the results are uniform in $j$. 
\end{proposition}

Then, we prove Lemma~\ref{lem:asymp_normal} as follows.
\begin{proof}
The KKT condition in \eqref{eq:DML-CATELasso} yields 
\begin{align*}
&\widehat{\Sigma} \left(\widehat{\bm{\beta}}^{\mathrm{WDML}}_n - \bm{\beta}_0\right)
+ \lambda \widehat{\bm{\kappa}}^{\mathrm{WDML}} =  \bm{W}^\top \widehat{\bm{\varepsilon}}^{\mathrm{WDML}} / n,
\end{align*}
where we used $\widehat{\mathbb{Q}}^{\mathrm{DML}} = \bm{W}\bm{\beta}_0 + \widehat{\bm{\varepsilon}}^{\mathrm{DML}}$, and 
\begin{align*}
&\bm{W}^\top\left(\widehat{\mathbb{Q}}^{\mathrm{WDML}}  - \bm{W}\widehat{\bm{\beta}}^{\mathrm{WDML}}_n\right)/n = \bm{W}^\top\left(\bm{X}\bm{\beta}_0 + \widehat{\bm{\varepsilon}}^{\mathrm{WDML}} - \bm{W}\widehat{\bm{\beta}}^{\mathrm{WDML}}_n\right) / n = \widehat{\Sigma}\left(\bm{\beta}_0 - \widehat{\bm{\beta}}^{\mathrm{WDML}}_n\right) + \bm{W}^\top\widehat{\bm{\varepsilon}}^{\mathrm{WDML}}/ n.
\end{align*}
Suppose that $\widehat{\Theta}$ is a reasonable approximation for the inverse of $\widehat{\Sigma}$. Then, it holds that
\begin{align*}
\widehat{\bm{\beta}}^{\mathrm{WDML}}_n - \bm{\beta}_0 &+ \widehat{\Theta} \lambda \widehat{\kappa}^{\mathrm{WDML}} = \widehat{\Theta} \bm{W}^\top \widehat{\bm{\varepsilon}}^{\mathrm{WDML}} / n - \widehat{\Delta}^{\mathrm{WDML}} / \sqrt{n},
\end{align*}
where $\widehat{\Delta}^{\mathrm{WDML}} \coloneqq \sqrt{n}\widehat{\Theta}\widehat{\Sigma} \left( \widehat{\bm{\beta}}^{\mathrm{WDML}}_n - \bm{\beta}_0\right)$. 

Furthermore, we defined $R$ as $R = \sqrt{n}\widehat{\Theta} \bm{W}^\top\Big(\widehat{\bm{\varepsilon}}^{\mathrm{WDML}} - \overline{\bm{\varepsilon}}^{\mathrm{Weight}}\Big)$. Therefore, we obtain 
\begin{align*}
\widehat{\bm{\beta}}^{\mathrm{WDML}}_n - \bm{\beta}_0 &+ \widehat{\Theta} \lambda \widehat{\kappa}^{\mathrm{WDML}} = \widehat{\Theta} \bm{W}^\top \overline{\bm{\varepsilon}}^{\mathrm{Weight}} / n - \widehat{\Delta}^{\mathrm{WDML}} / \sqrt{n} - R,
\end{align*}

Here, we can show that $\left\|\widehat{\Delta}^{\mathrm{WDML}}\right\|_1 = o_P(1)$ and $\|R\|_1 = o_P(1)$ hold.

This result yields the debiased Lasso estimator, defined as
\[\widehat{\bm{b}}^{\mathrm{WTDL}}_n = \widehat{\bm{\beta}}^{\mathrm{WDML}}_n + \widehat{\Theta} \bm{W}^\top\left(\widehat{\mathbb{Q}}^{\mathrm{WDML}} - \bm{W}\widehat{\bm{\beta}}^{\mathrm{WDML}}_n\right) / n.\]
Therefore, we define the debiased DML-CATELasso estimator as $\widehat{\bm{b}}^{\mathrm{WTDL}} \coloneqq \overline{\bm{\beta}}^{\mathrm{WDML}}_n + \widehat{\Theta}  \lambda\widehat{\bm{\kappa}}$, which is equal to
\begin{align*}
\widehat{\bm{b}}^{\mathrm{WTDL}}
&= \widehat{\bm{\beta}}^{\mathrm{WDML}}_n + \widehat{\Theta} \bm{W}^\top\left(\widehat{\mathbb{Q}}^{\mathrm{WDML}} - \bm{W}\widehat{\bm{\beta}}^{\mathrm{WDML}}_n\right) / n,
\end{align*}
where we used \eqref{eq:KKTcond}. 
Here, it holds that 
\begin{align*}
\sqrt{n}\left(\widehat{\bm{b}}^{\mathrm{WTDL}} - \bm{\beta}_0\right) = \widehat{\Theta} \bm{W}^\top \overline{\bm{\varepsilon}}^{\mathrm{Weight}} / \sqrt{n} - \widehat{\Delta}^{\mathrm{WDML}},
\end{align*}
where $\widehat{\Theta} \bm{X}^\top \overline{\bm{\varepsilon}}^{\mathrm{Weight}} / \sqrt{n}$ converges to a Gaussian distribution from the central limit theorem and $\Delta$ converges to zero as $n\to \infty$. Lastly, we consider how the debiased Oracle-DR-CATELasso estimator $\overline{\bm{b}}$ relates to the debiased Lasso for WDML-CATELasso estimator. Based on the above arguments, we define the debiased Lasso estimator based on the WDML-CATELasso estimator as
\begin{align*}
\widehat{\bm{b}}^{\mathrm{WTDL}}_n
&= \widehat{\bm{\beta}}^{\mathrm{WDML}}_n + \widehat{\Theta} \bm{W}^\top\left(\widehat{\mathbb{Q}}^{\mathrm{WDML}} - \bm{W}\widehat{\bm{\beta}}^{\mathrm{WDML}}_n\right) / n.
\end{align*}
We refer to this estimator as the TDL estimator since we apply the DML and debiased Lasso together. Here, it holds that 
\begin{align*}
\sqrt{n}\left(\widehat{\bm{b}}^{\mathrm{WTDL}}_n - \bm{\beta}_0\right) = \widehat{\Theta} \bm{W}^\top \overline{\bm{\varepsilon}}^{\mathrm{Weight}} / \sqrt{n} - \widehat{\Delta}^{\mathrm{WDML}} + o_P(1),
\end{align*}
The detailed statement of this results are shown in Lemma~\ref{lem:conv_Q} with its proof. 
The term $\widehat{\Theta} \lambda \widehat{\bm{\kappa}} $ debiases the orignal Lasso estimator. The term comes from the KKT conditions. This term is similar to one in the debiased Lasso in \citet{vandeGeer2014} but different from it because we penalizes the difference $\bm{\beta}$.
As shown in Section~\ref{sec:theoretical}, we can develop the confidence intervals for the debiased CATE Lasso estimator if the inverses $\widehat{\Theta}$ is appropriately approximated. 

    From Proposition~\ref{lem:random}, it holds that 
    \begin{align*}
        \frac{1}{\widehat{\tau}^2_j} = \frac{1}{\tau^2_j + o_P(1)}.
    \end{align*}
    From the assumption, it holds that $\frac{1}{\widehat{\tau}^2_j} = o_P(1)$. Therefore, from Theorem~\ref{thm:nongaussian}, it holds that
    \begin{align*}
        \left\|\widehat{\Delta}^{\mathrm{WDML}}\right\|_1 = o_P(s_0\log(p)/\sqrt{n}) = o_P(1).
    \end{align*}
\end{proof}

\end{document}